\documentclass{article}
\usepackage[]{graphicx}\usepackage[]{color}
\makeatletter
\def\maxwidth{ %
  \ifdim\Gin@nat@width>\linewidth
    \linewidth
  \else
    \Gin@nat@width
  \fi
}
\makeatother

\definecolor{fgcolor}{rgb}{0.345, 0.345, 0.345}
\newcommand{\hlnum}[1]{\textcolor[rgb]{0.686,0.059,0.569}{#1}}%
\newcommand{\hlstr}[1]{\textcolor[rgb]{0.192,0.494,0.8}{#1}}%
\newcommand{\hlcom}[1]{\textcolor[rgb]{0.678,0.584,0.686}{\textit{#1}}}%
\newcommand{\hlopt}[1]{\textcolor[rgb]{0,0,0}{#1}}%
\newcommand{\hlstd}[1]{\textcolor[rgb]{0.345,0.345,0.345}{#1}}%
\newcommand{\hlkwa}[1]{\textcolor[rgb]{0.161,0.373,0.58}{\textbf{#1}}}%
\newcommand{\hlkwb}[1]{\textcolor[rgb]{0.69,0.353,0.396}{#1}}%
\newcommand{\hlkwc}[1]{\textcolor[rgb]{0.333,0.667,0.333}{#1}}%
\newcommand{\hlkwd}[1]{\textcolor[rgb]{0.737,0.353,0.396}{\textbf{#1}}}%

\usepackage{framed}
\makeatletter
\newenvironment{kframe}{%
 \def\at@end@of@kframe{}%
 \ifinner\ifhmode%
  \def\at@end@of@kframe{\end{minipage}}%
  \begin{minipage}{\columnwidth}%
 \fi\fi%
 \def\FrameCommand##1{\hskip\@totalleftmargin \hskip-\fboxsep
 \colorbox{shadecolor}{##1}\hskip-\fboxsep
     \hskip-\linewidth \hskip-\@totalleftmargin \hskip\columnwidth}%
 \MakeFramed {\advance\hsize-\width
   \@totalleftmargin\z@ \linewidth\hsize
   \@setminipage}}%
 {\par\unskip\endMakeFramed%
 \at@end@of@kframe}
\makeatother

\definecolor{shadecolor}{rgb}{.97, .97, .97}
\definecolor{messagecolor}{rgb}{0, 0, 0}
\definecolor{warningcolor}{rgb}{1, 0, 1}
\definecolor{errorcolor}{rgb}{1, 0, 0}
\newenvironment{knitrout}{}{} 

\usepackage{alltt}
\usepackage[margin=1in]{geometry}
\usepackage{amsmath, amsthm, amssymb}
\usepackage{gensymb} 
\usepackage{caption} 
\usepackage{mathtools} 
\usepackage{mathtools} 
\usepackage{thmtools, thm-restate}
\usepackage[authoryear]{natbib}
\usepackage{authblk}

\DeclareMathOperator{\Dist}{Dist}

\DeclareMathOperator{\var}{var}
\DeclareMathOperator{\cov}{cov}
\DeclareMathOperator*{\argmax}{arg\,max}
\DeclareMathOperator*{\argmin}{arg\,min}
\DeclareMathOperator{\Unif}{Unif}
\DeclarePairedDelimiter{\ceil}{\lceil}{\rceil}
\DeclarePairedDelimiter{\floor}{\lfloor}{\rfloor}

\newcommand{\Beta}{{\boldsymbol\beta}}%
\newcommand{\B}{\mathcal{B}}
\newcommand{\e}{\mathbf{e}}
\newcommand{\E}{\mathbf{E}}
\newcommand{\Eta}{{\boldsymbol\eta}}%
\newcommand{\inv}{^{-1}}
\newcommand{\LL}{\mathbf{L}}
\newcommand{\N}{\mathcal{N}}
\newcommand{\p}{\mathbf{P}}
\newcommand{\R}{\mathbf{R}}
\newcommand{\RR}{\mathbb{R}}
\newcommand{\x}{\mathbf{x}}%
\newcommand{\X}{\mathbf{X}}
\newcommand{\XI}{{\boldsymbol\xi}}
\newcommand{\Y}{\mathbf{Y}}
\newcommand{\z}{\mathbf{z}}

\newcommand{\as}{\stackrel{a.s.}{\rightarrow}}
\newcommand{\ip}{\stackrel{P}{\rightarrow}}

\newcommand\numberthis{\addtocounter{equation}{1}\tag{\theequation}}

\theoremstyle{plain}
\newtheorem{thm}{Theorem}[section]
\newtheorem{lem}[thm]{Lemma}

\theoremstyle{definition}

\theoremstyle{remark}

\title{Regression of ranked responses when raw responses are censored}
\author[1]{Michael C. Donohue}
\author[2,3]{Anthony C. Gamst}
\author[3]{Robert A. Rissman}
\author[4]{Ian Abramson}
\affil[1]{Alzheimer's Therapeutic Research Institute, University of Southern California}
\affil[2]{Division of Biostatistics \& Bioinformatics, University of California, San Diego}
\affil[3]{Department of Neurosciences, University of California, San Diego}
\affil[4]{Department of Mathematics, University of California, San Diego}

\IfFileExists{upquote.sty}{\usepackage{upquote}}{}
\begin{document}
\maketitle

\begin{abstract}
We discuss semiparametric regression when only the ranks of responses are observed.
The model is $Y_i = F (\x_i'\Beta_0 + \varepsilon_i)$,
where $Y_i$ is the unobserved response, $F$ is a monotone increasing function,
$\x_i$ is a known $p-$vector of covariates, $\Beta_0$ is an unknown
$p$-vector of interest, and $\varepsilon_i$ is an error term independent of
$\x_i$. We observe $\{(\x_i,R_n(Y_i)) : i = 1,\ldots ,n\}$,
where $R_n$ is the ordinal rank function. We explore a novel
estimator under Gaussian assumptions. We discuss the literature, apply the method 
to an Alzheimer's disease biomarker, conduct simulation studies, and prove consistency 
and asymptotic normality.
\end{abstract}


\section{Introduction}

Rank-based statistics are often attractive for their
robustness properties. Occasionally, due to some practical measurement
difficulties, we have no choice but to resort to ranks. For example, we might
wish to analyze webpage or team ranks without access to an underlying continuous response
variable. Suppose data arises from a monotone transformation of a linear model. How much
information is lost when we observe ranks in place of raw data and how
well can we estimate the linear parameters of the regression? One
might think that the rank transformation, as depicted in Figure \ref{fig:1},
causes a catastrophic loss of information about the target parameter.
It is clear that the scale of the linear parameter and any intercept
term are irrecoverable, however we can estimate the ``direction'' of the
parameter (i.e. up to a scalar). Such an estimate can be useful for inference regarding the
relative importance of effects, predicting ranks, and semiparametric 
estimation of response surfaces with \emph{parallel} linear linear level sets. 
Response surfaces with linear (possibly non-parallel) level sets exhibit ``\emph{asynergy};'' and synergy can be assessed by examination of the residuals of asynergistic fits \citep{donohue_asynergistic_2007}.

Semiparametric rank-based estimators introduced by \cite{han_non-parametric_1987}
and \cite{sherman_limiting_1993} actually yield $\sqrt{n}$-consistent and asymptotically
normal estimates for the direction of the parameter competitive with
ordinary least squares methods using the raw observations. These rank
based estimators all involve maximizing some form of rank correlation,
and can be computationally complex in higher dimensions. We will explore another
asymptotically normal estimate which is admittedly less robust than
these estimators, but which is achieved by a simpler ordinary least
squares computation. We prove consistency and asymptotic normality under
strong Gaussian assumptions. Simulation studies demonstrate the methods sensitivity
to the Gaussian assumptions.

We also apply the method and competitors to a dataset with an Alzheimer's disease blood plasma assay 
that is prone to batch effects \citep{donohue_longitudinal_2014}. 
Bioassay florescence intensities are typically 
parametrically calibrated, plate by plate, using standards of know concentration 
(e.g. \cite{davidian_regression_1990}). We explore the use of the rank transformation, also 
applied plate by plate, as an alternative nonparametric standardization under the 
assumption that there is negligible biological variation across plates. We then 
apply the rank-based regression methods under discussion.

\begin{figure}
\begin{knitrout}
\definecolor{shadecolor}{rgb}{0.969, 0.969, 0.969}\color{fgcolor}
\includegraphics[width=\linewidth]{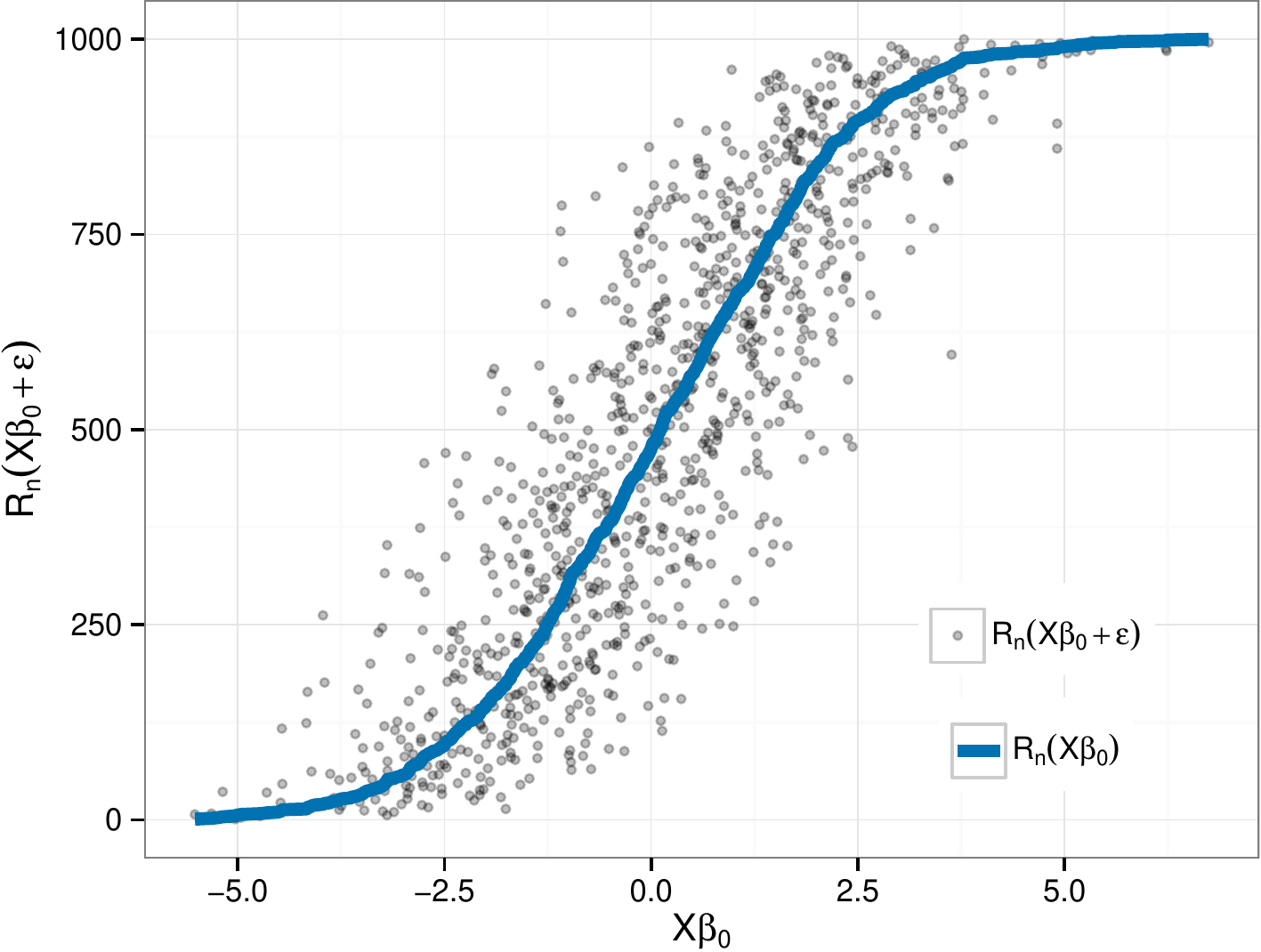} 

\end{knitrout}
\caption{The rank transformation. Rank transformed responses (grey dots) retain information about the parameter of interest, $\Beta_0$.}\label{fig:1}
\end{figure}

\subsection{A novel gaussian quantile rank-based regression} 
Consider the restricted rank-based regression problem under assumptions
\begin{itemize}
  \item[(A1)] $\x_i \sim \N_p(0, \Sigma)$
  \item[(A2)] $\varepsilon_i  \sim \N(0, \sigma^2)$ 
  \item[(A3)] $Y_i = F (\x_i'\Beta + \varepsilon)$, $F$ is monotone
increasing
  \item[(A4)] We observe $\z_i = (\x_i, R_n(Y_i))$ for $i = 1, \dots , n$
\end{itemize}
However, we can assume without loss of generality that $F$
is the identity map since $F$ has no affect on the observed ranks and it
is a nuisance parameter. Furthermore let $G$ denote the distribution of $\x_i$ and $H$ denote the
distribution for $Y_i$ (both of which are Gaussian). Let $H_n$ be
the empirical distribution function based on a sample $Y_1 , \dots ,
Y_n$. We seek a competing estimator to Spearmax \eqref{eq:spearmax}, still based
only on the observation $\{(\x_i, R_n(Y_i)) : i = 1 \dots n\}$,
that takes advantage of two facts:
\begin{equation}
  n\inv R_n(Y_i) = H_n (Y_i)
\end{equation}
and
\begin{equation}\label{eq:keyobs}
 \Phi\inv (H(Y_i)) = c(\x_i'\Beta_0 + \varepsilon_i)
\end{equation}
for some constant, $c$, where $\Phi\inv$ is the standard
normal quantile function (see Lemma \ref{lemone} in Section \ref{theory}). In light of these facts, we propose the Gaussian Quantile Regression: 
\begin{equation} \label{eq:gqr}
\hat \Beta_n = \Big(n\inv \sum_{i=1}^n \x_i'\x_i\Big)\inv
n\inv \sum_{i=1}^n \x_i'\Phi\inv(H_n^*(Y_i)),
\end{equation}
where $H_n^* (y) = (n + 1)\inv R_n(y)$. In the
appendix we show consistency of $\Beta_n$ and asymptotic normality
of a version of $\hat \Beta_n$ where $\Phi\inv$ is replaced with a truncated
quantile function, $\Phi\inv_n$. 

In the next section we will discuss related background literature. In Section 3 we discuss consistency and asymptotic normality. In Section 4 we describe simulations.

\section{Background}
\subsection{$R$-Estimators}
Note the distinction between our setting and that of classical ``rank regression.''
\emph{Rank estimators} ($R$-\emph{estimators}) \cite{jureckova_robust_1996}, 
are not useful in our setting since they require observations of the response.
They utilize the ranks of the \emph{residuals}, not the ranks of the responses. 
More specifically, $R$-estimators solve $\argmin_{\Beta\in\B} \|\LL_n(\Beta)\|$, where
\[
\LL_n(\Eta) = \sum_{i=1}^n(\x_i - \bar\x_n)a_n(R_n(Y_i - \x_i\Eta)),
\]
and $a_n$ is an appropriate score function. \citet{parzen_resampling_1994} 
introduced a resampling method for inference regarding $\Beta_0$
using $R$-estimates.

\subsection{Monotonic linear index models and ``Spearmax''}
\citet{han_non-parametric_1987} introduced a 
``semiparametric monotonic linear index model'' of the form 
\begin{equation} \label{eq:han}
Y_i = D \circ F(\x_i'\Beta_0,\varepsilon_i)
\end{equation}
where $Y_i$ is an
observed response, $D$ is a monotone increasing function, $F$ is strictly
increasing in both arguments, $\x_i$ is a known $p$-vector, $\Beta_0$ is an
unknown $p$-vector of interest, and $\varepsilon_i$ is random error independent of
$\x_i$. Han also introduced a maximum rank correlation estimator
\begin{equation} \label{eq:2}
\argmax_{\Beta \in \B} \frac{1}{n(n-1)}\sum_{i\neq j}\{Y_i > Y_j\}\{\x_i'\Beta >
\x_j'\Beta\}
\end{equation}
where $\{\cdot\}$ represents the
indicator function and $\B$ is an appropriate subset of $\RR^p$, the unit
ball say. The necessity of restricting to $\RR^p$ becomes apparent in
light of the fact that scalar multiples of a particular $\Beta$ yield the same
rank correlation. The power of estimators of this type is that they
exploit monotonicity without making assumptions about the particular
form of $D$ or $F$ . Assumptions
regarding the error distribution are minimal as well. \cite{sherman_limiting_1993}
showed \eqref{eq:2} is $n$-consistent and asymptotically normal. \citet{cavanagh_rank_1998} proposed a class of consistent and asymptotically
normal estimators of the form
\begin{equation} \label{eq:cav}
\hat\Beta_n = \argmax_{\Beta \in \B} \sum_{i=1}^n M(Y_i)R_n(\x_i'\Beta)
\end{equation}
where $R_n$ is the rank
function 
\begin{equation} 
R_n(a_i) = \sum_{j=1}^n\{a_j \leq a_i\}\textrm{, for }
(a_1 ,\ldots,a_n ) \in \RR^n
\end{equation}
and $M$ is either
deterministic or $R_n$. When $M$ is $R_n$, the quantity being maximized
in \eqref{eq:cav} is a linear function of Spearman's rank correlation
coefficient, and in this case we need only observe the ranks, not the
raw responses. When $M = R_n$ we will refer to \eqref{eq:cav} as the \emph{Spearmax}
estimate of $\Beta_0$:
\begin{equation} \label{eq:spearmax}
\hat\Beta_n = \argmax_{\Beta \in \B} \sum_{i=1}^n R_n(Y_i)R_n(\x_i'\Beta)
\end{equation}
The problem of
estimating $\Beta_0$ becomes one of maximizing a step function over $\B$.

\subsection{$M$-Estimators}
  The Spearmax estimator is related to the
classical $M$\emph{-estimator} \cite{jureckova_robust_1996} for the linear
regression model:
\[
\argmin_{\Beta \in \RR^p} \sum_{i=1}^n \rho(Y_i - \x_i'\Beta)
\]
for an appropriate absolutely continuous and differentiable $\rho: \RR \rightarrow \RR$. The Spearmax estimate can be written as a related minimization 
\begin{align*}
\argmax_{\Beta \in \B} \sum_{i=1}^n R_n(Y_i)R_n(\x_i'\Beta)
  = & \argmin_{\Beta \in \B} \|\R_n(\Y)\|^2 + \|\R_n(\X\Beta)\|^2 - 2\sum_{i=1}^n R_n(Y_i)R_n(\x_i'\Beta)\\
  = & \argmin_{\Beta \in \B} \|\R_n(\Y) - \R_n(\X\Beta)\|^2 \\
  = & \argmin_{\Beta \in \B} \sum_{i=1}^n (R_n(Y_i) - R_n(\x_i'\Beta))^2
\end{align*}
However, the rank functions involved
make Spearmax distinct from classical $M$-estimation.

\subsection{Current Status Data}
Another interesting and related problem is
\emph{rank-based regression for current status data} \cite{aragon_rank_1995,
abrevaya_rank_1999} with survival analysis applications. The
setting is also linear regression with the response $Y_i$ not
observable. We observe $(X_i, C_i, \Delta_i)$, where $\Delta_i$ denotes
the indicator on $\{Y_i \leq C_i\}$. \citet{aragon_rank_1995} 
proposed and showed consistency of
\begin{equation} \label{eq:aragon}
\argmax_{\Beta\in\B}\sum_{i=1}^n \Delta_iR_n(C_i - \x_i'\Beta)
\end{equation}
and
\[
\argmax_{\Beta\in\B}\sum_{i=1}^n R_n(C_i - \x_i'\Beta)\hat F_\Beta(C_i - \x_i'\Beta)
\]
where $\hat F_\Beta$ is a uniform
strongly consistent estimate of the distribution of $(Y_i - \x_i'\Beta)$.
\citet{abrevaya_rank_1999} proved asymptotic normality of \eqref{eq:aragon}.

\subsection{Concomitants}
Statistics related to our Gaussian Quantile Regression \eqref{eq:gqr}
appear in the concomitant literature. 
\citet{david_18_1998} provide a detailed review. We will summarize these
results and provide a comparison to our proposed estimator. 

\citet{yang_linear_1981} discussed related statistics of a more general form. Let
$(X_i, Y_i), i = 1, \dots , n$ be a random sample from a bivariate
distribution with cumulative distribution function $F(x, y)$. Let
$Y_{i:n}$ denote the $i$th order statistic and $X_{[i:n]}$ denote the
the so-called $i$th concomitant, that is the $X$ variable associated with
$Y_{i:n}$. The term \emph{induced order statistics} has also been used in place
of concomitant by \cite{bhattacharya_convergence_1974}. Under mild conditions, Yang
proved asymptotic normality of statistics of the form
\begin{equation}\label{eq:yang}
   S_n = n\inv \sum_{i=1}^n J(i/(n + 1))K(X_{[i:n]} , Y_{i:n}),
\end{equation}
where $J$ is a smooth bounded function (possibly depending on $n$) and $K$ is some
real valued function on $\RR^2$. We can rewrite our Gaussian Quantile Regression \eqref{eq:gqr} as
\[
\hat \Beta_n = \Big(n\inv \sum_{i=1}^n \x_i'\x_i\Big)\inv
n\inv \sum_{i=1}^n \x_{[i:n]}'\Phi\inv(i/(n+1)),
\]
where we generalize the definition of $i$th concomitant 
here to be the $\x$ vector associated with
$Y_{i:n}$. In the notation of Yang's statistics we have $J = \Phi\inv$ and
$K(\x_{[i:n]} , Y_{i:n} ) = \x_{[i:n]}$. Yang considers 
estimators of $\E(Y |X = x)$, $\p(Y \in A|X = x)$ and $\var(Y
|X = x)$ which include observations of raw $Y$ values to which
we are not privy. Yang also only discusses the case $p = 1$. The method
of proof used by Yang is based on earlier methods of \citet{stigler_linear_1969}
and \citet{hajek_asymptotic_1968}, namely H\'{a}jek's projection lemma (see also
\citet{hettmansperger_statistical_1984} page 50). The idea is to show 
$S_n$ \eqref{eq:yang} and
its projection 
\[
\hat S_n = \sum_{i=1}^n \E(S_n |\x_i, Y_i)-(n-1)\E(S_n)
\]
are asymptotically equivalent in terms of
mean square. The asymptotic normality of the projection follows from
the central limit theorem since the summands are i.i.d. We will deploy a
similar strategy of approximating the dependent sum by an independent
sum and invoking the central limit theorem. However, we will see that
our approximation is of the form:
\[
 n\inv \sum_{i=1}^n\x_i\Phi\inv(H_n^*(Y_i))\approx
 n\inv \sum_{i=1}^n\x_i\Phi\inv(H_n(Y_i)).
\]
The result is also i.i.d. summands, but in our regression
setting, this approximation allows the target regression parameter to
emerge via the observation \eqref{eq:keyobs}. 

\citet{yang_linear_1981-1} also proved
asymptotic normality of
\begin{equation}\label{eq:yang2}
\sum_{i=1}^n J(t_{ni})X_{[i:n]},
\max_{1\leq i \leq n}|t_{ni}- i/n| \rightarrow 0
\end{equation}
for some deterministic double indexed sequence $\{t_{ni}\}$. Statistics of this
form have applications to tests for normality and independence. It is
interesting to note that in an application of the rank-based
regression, a statistic such as \eqref{eq:yang2} would be useful since it offers a
test for the Gaussian assumption without raw response observations.

\subsection{$L$-Estimators and $\alpha$-Regression Quantiles}
Related to our proposed rank-based regression and the concomitant statistics such as \eqref{eq:yang2} are
the so called $L$-estimators (linear combinations of functions of order
statistics) \cite{jureckova_robust_1996}. In the $p = 1$ location
setting, these are statistics of the form 
\[
L_n = n\inv\sum_{i=1}^n X_{[n:i]} J_n (i/(n + 1)).
\]
$L$-estimators were first adapted to the linear regression problem with
general $p$ in 1973 by \citet{bickel_analogues_1973}. \citet{koenker_regression_1978}
introduced $\alpha$\emph{-regression quantiles} $\hat\Beta(\alpha)$, $0 < \alpha < 1$
as the solution to 
\begin{equation*}
\hat\Beta = \argmin_{\Beta\in\B} \sum_{i=1}^n\rho_\alpha(Y_i - \x_i\Beta),
\end{equation*}
\begin{equation*}
\rho_\alpha(x) = |x| ((1 - \alpha)\{x < 0\} + \alpha\{x > 0\}), 
\end{equation*}
where $\{A\}$ is the indicator on $A$. Koenker and
Basset also posed the problem as a linear program and derived the
asymptotic distribution. 

\section{Consistency and asymptotical normality}\label{theory}

\subsection{Consistency} 
In this section we will demonstrate
$\Phi\inv (H(Y_i)) = c(\x_i'\Beta_0 + \varepsilon_i)$ and the consistency of 
$\hat\Beta_n$ \eqref{eq:gqr}. Table \ref{tab:tab1} collects the notation and assumptions. Again, we assume without loss of generality that $D$
is the identity map. 

The first lemma is an impetus for the estimator
since the target parameter appears from a composition of a quantile
function and a distribution function. The proof of the lemma
relies on the assumption
that all the random variables are Gaussian or a linear combination of
Gaussian random variables. 

\begin{lem}\label{lemone}
\[
\Phi\inv(H(Y_i)) = \sigma^*(\x_i'\Beta_0 + \varepsilon_i)
\]
where $\sigma^* = (\sigma_{\Beta_0}^2+\sigma^2)^{-1/2}$, 
$\sigma_{\Beta_0}^2=\Beta_0'\Sigma\Beta_0$.
\end{lem}

\begin{proof}
\begin{align*}
H(Y_i) & = \p[Y_j \leq Y_i] \tag*{$i\neq j$}\\
& = \p\left[ 
    \frac{\x_j\Beta_0+\varepsilon_j}{\sqrt{\sigma_{\Beta_0}^2+\sigma^2}} \leq
    \frac{\x_i\Beta_0+\varepsilon_i}{\sqrt{\sigma_{\Beta_0}^2+\sigma^2}}
  \right]\\
& = \Phi\left(\frac{\x_i\Beta_0+\varepsilon_i}{\sqrt{\sigma_{\Beta_0}^2+\sigma^2}}\right).\\
\end{align*}
Now observe 
\begin{equation*}
\Phi\inv(H(Y_i)) = \frac{\x_i\Beta_0+\varepsilon_i}{\sqrt{\sigma_{\Beta_0}^2+\sigma^2}}.
\end{equation*}
\end{proof}

The following series of lemmas will be used to show $L_1$ convergence of
the summands, $\xi_{nij}$, and the fact that $\cov(\XI_{n1j} , \XI_{n2j})\rightarrow 0$ (proofs are in the appendix). These results and Chebychev's Inequality will be enough to yield consistency.

\begin{restatable}{lem}{lemtwo}
\label{lemtwo}
$\xi_{nij}\as\xi_{ij}$ as $n\rightarrow\infty$
\end{restatable}

\begin{minipage}{.75\textwidth}
\hrulefill

\centering Assumptions and Notation

\hrulefill
\begin{itemize}
\item[(A1)] $\x_i \sim \N_p(0, \Sigma)$, $\Sigma$ is nonsingular
\item[(A2)] $\varepsilon_i  \sim \N(0, \sigma^2)$
\item[(A3)] $Y_i = F (\x_i'\Beta + \varepsilon)$, $F$ is monotone
increasing
\item[(A4)] We observe $\z_i = (\x_i, R_n(Y_i))$ for $i = 1, \dots , n$
\end{itemize}
\begin{align*}
\x_i & \sim G \textrm{ for some distribution } G\\
Y_i & \sim H \textrm{ for some distribution } H\\
\sigma_{\Beta_0}^2 & = \Beta_0'\Sigma\Beta_0\\
\sigma^* & = (\sigma_{\Beta_0}^2+\sigma^2)^{-1/2}\\
R_n(y) & = \sum_{j=1}^n\{Y_j\leq y\}\\
H_n(y) & = n\inv R_n(y)\\
H_n^*(y) & = (n+1)\inv R_n(y)\\
\XI_i & = \x_i\Phi\inv(H(Y_i))\\
\xi_{ij} & = X_{ij}\Phi\inv(H(Y_i))\\
\XI_{ni} & = \x_i\Phi\inv(H_n^*(Y_i))\\
\xi_{nij} & = X_{ij}\Phi\inv(H_n^*(Y_i))\\
\hat\Beta_n & = (n\inv \sum_{i=1}^n \x_i'\x_i)\inv n\inv \sum_{i=1}^n \x_i'\Phi\inv(H_n^*(Y_i))
\end{align*}

\hrulefill

\captionof{table}{Gaussian Quantile Regression assumptions and notation.}
\label{tab:tab1}            
\end{minipage}

\begin{restatable}{lem}{lemthree}
\label{lemthree}  
$H_n^* (Y_i) \sim U_n$ where $U_n \sim \Unif\{(n + 1)\inv, \dots, n(n + 1)\inv \}$. 
\end{restatable}

\begin{restatable}{lem}{lemfour}
\label{lemfour}
  $\E\Phi\inv(U_n)^4 \leq 6$.
\end{restatable}

\begin{restatable}{lem}{lemfive}
\label{lemfive}
$\E\xi^2_{nij} \leq 3\sqrt{2}\sigma_j^2$, where $X_{ij}\sim \N(0, \sigma^2_j)$. 
\end{restatable}

\begin{restatable}{lem}{lemsix}
\label{lemsix}
$\E\xi_{nij}\rightarrow \E\xi_{ij}$ as $n\rightarrow\infty$
\end{restatable}

\begin{restatable}{lem}{lemseven}
\label{lemseven}
$\cov(\xi_{nsj}, \xi_{ntj})\rightarrow 0$ as $n\rightarrow\infty$.
\end{restatable}

\begin{thm}
The estimate $\hat\Beta_n$ is consistent for a scalar factor of $\Beta_0$.
\end{thm}
\begin{proof}
Let $\bar\xi_{nj}= n\inv \sum_{i=1}^n\xi_{nij}$. By
Chebychev's Inequality:
\begin{align*}
\p [|\bar\xi_{nj}-\E\xi_{nij}|\geq \delta] &\leq \var(\bar\xi_{nj})/\delta^2\\
& = (n\delta)^{-2}(n \var(\xi_{nij}) + n(n - 1) \cov(\xi_{nsj},\xi_{ntj}))\\
& \leq \delta^{-2}(n\inv\E\xi_{nij}^2 + n\inv (n - 1)\cov(\xi_{nsj},\xi_{ntj}))\\
& \leq \delta^{-2}(n\inv3\sqrt{2}\sigma_j^2 + n\inv(n-1)\cov(\xi_{nsj},\xi_{ntj}))\\
&\rightarrow 0
\end{align*}
So for each $j\in \{1, \dots, p\}$
by the above and Lemma \ref{lemsix} we have, 
\[
 |\bar\xi_{nj}- \E\xi_{ij}| \leq 
 |\bar\xi_{nj}-\E\xi_{nij}|+|\E\xi_{ij} - \E\xi_{nij}| \ip 0,
\]
and therefore
\[
 \bar\XI_n \ip \E\XI_i = 
 \E\left[\frac{\x_i(\x_i'\Beta_0 + \varepsilon_i)}{\sqrt{\sigma_{\Beta_0}^2 + \sigma^2}}\right]
 = \sigma_*\Sigma\Beta_0
\]
We also have, as
with the classical ordinary least squares setup, 
\[
\left(n\inv\sum_{i=1}^n \x_i'\x_i\right)\inv \ip \Sigma\inv.
\]
Finally we have: 
\[
\hat\Beta_n = \left(n\inv\sum_{i=1}^n \x_i'\x_i\right)\inv \bar\XI_n \ip \sigma_*\Beta_0.
\]
\end{proof}

\subsection{Asymptotic normality}
In order to prove asymptotic normality of our estimator we are forced to truncate the
Gaussian quantile function. For the extreme ranked data
$(R_n(Y_i) \notin (1 - \alpha_n ,\alpha_n ))$, this results in
$\x_i\Phi\inv(H_n (Y_i))$ being replaced with $\x_i\Phi\inv(1 - \alpha_n)$
or $\x_i\Phi\inv(\alpha_n)$ where
\[
\alpha_n = \Phi\left(\sqrt{\frac{1}{2}\log n}\right) \rightarrow 1
\]

More precisely, we define the truncated version of the Gaussian quantile
function as 

\begin{equation}
  \Phi_n\inv(x)=\begin{cases}
    \Phi\inv(1-\alpha_n), & \text{if } x\in(0,1-\alpha_n]\\
    \Phi\inv(x),          & \text{if } x\in(1-\alpha_n, \alpha_n]\\
    \Phi\inv(\alpha_n),   & \text{if } x\in[\alpha_n,1)
  \end{cases}
\end{equation}
We will demonstrate that truncating at this particular
$\alpha_n$ allows us to put a $cn^{1/4}$ bound on the first derivative of
the quantile function. We prove asymptotic normality of the truncated
version of the estimator \eqref{eq:gqr}
\begin{equation}\label{eq:tgqr}
  \tilde \Beta_n = \left(n\inv \sum_{i=1}^n \x_i'\x_i \right)\inv
   n\inv \sum_{i=1}^n \x_i'\Phi\inv_n(H_n^*(Y_i))
\end{equation}
by applying the Central Limit Theorem to
an i.i.d. version of our estimator, then applying Slutsky's Theorem. The
three series in discussion are
\begin{align*}
  S_n^{(2)} & = \sqrt{n}\left( n\inv\sum_{i=1}^n \x_i'\Phi\inv(H(Y_i)) - \sigma_*\Sigma\Beta_0\right) \Rightarrow \N_p(0, A), \tag{CLT}\\
  S_n^{(1)} & = \sqrt{n}\left( n\inv\sum_{i=1}^n \x_i'\Phi\inv_n(H(Y_i)) - \sigma_*\Sigma\Beta_0\right)\textrm{, and}\\
  S_n       & = \sqrt{n}\left( n\inv\sum_{i=1}^n \x_i'\Phi\inv_n(H_n^*(Y_i)) - \sigma_*\Sigma\Beta_0\right).
\end{align*}
In the first line above (CLT), $A$ is
the $p \times p$ dispersion matrix associated with
\[
x_1' \Phi\inv (H(Y_1 )) = \sigma_* x_1'(x_1'\Beta_0 + \varepsilon_1)
\]
 with entries
\begin{align*}
  a_{ij} = & \cov(\sigma_*X_{1i}(x_1' \Beta_0 + \varepsilon_1), \sigma_*X_{1j}(\x_1' \Beta_0 + \varepsilon_1))\\
   = & \sigma_*^2 (\E(X_{1i}(\x_1 \Beta_0 + \varepsilon_1)X_{1j}(\x_1' \Beta_0 + \varepsilon_1 ))\\
      & -\E(X_{1i}(\x_1' \Beta_0 + \varepsilon _1))\E(X_{1j}(\x_1' \Beta_0 + \varepsilon_1)))\\
  = & \sigma_*^2(\E X_{1i}X_{1j}((\x_1' \Beta_0)^2 + \sigma^2) - (\e_i'\Sigma\Beta_0)(\e_j'\Sigma\Beta_0))\\
  = & \frac{1}{\sigma_{\Beta_0}^2+\sigma^2}(\E X_{1i}X_{1j}(\x_1'\Beta_0)^2 + \sigma^2\e_i'\Sigma\e_j - (\e_i'\Sigma\Beta_0)(\e_j'\Sigma\Beta_0)) \numberthis \label{eq:15}
\end{align*}
Here $\e_i$ denotes the $p$-vector with 1 as its $i$th component and 0's
elsewhere. 

We show $\|S^{(2)}_n - S^{(1)}_n\| \ip 0$, $\|S^{(1)}_n - S_n\| \ip 0$, and apply Slutsky's Theorem to conclude $S_n$ is asymptotically normal (proofs of lemmas are in the appendix). 

\begin{restatable}{lem}{lemeight}
\label{lemeight}
$\|S^{(2)}_n - S^{(1)}_n\| \ip 0$
\end{restatable}

We will need two lemmas to prove $\|S^{(1)}_n - S_n\|\ip 0$.

\begin{restatable}{lem}{lemnine}
\label{lemnine}
\[
(\Phi\inv_n(H(Y_1)) - \Phi\inv_n(H_n(Y_1)))^2 \leq \sigma_*2\pi\sqrt{n}(H(Y_1) - H_n^*(Y_1))^2
\]
\end{restatable}

\begin{restatable}{lem}{lemten}
\label{lemten}
\[
\E_1 (H(Y_1) - H_n^*(Y_1))^2 \leq \frac{1}{n + 1}
\]
almost surely where $\E_1$ is the conditional expectation $\E_1 (\cdot) =
\E(\cdot|Y_1 , \x_1)$.
\end{restatable}

\begin{restatable}{lem}{lemeleven}
\label{lemeleven}
\[
\|S_n^{(1)} - S_n \| \ip 0
\]
\end{restatable}

\begin{thm}
\[
S_n \Rightarrow \N_p(0,A) 
\]
\end{thm}
\begin{proof}
From the preceding lemmas we have
\[
\|S_n - S^{(2)}_n\| \leq \|S_n - S^{(1)}_n\| + \| S_n^{(1)} - S^{(2)}_n\| \ip 0
\]
Since $S^{(2)}_n \Rightarrow \N_p(0,A)$ by the Central Limit
Theorem, Slutsky's Theorem gives us $S_n\Rightarrow \N_p(0,A)$. 
\end{proof}
	
Finally we have asymptotic normality of our estimate in the following sense (proofs are in the appendix).

\begin{restatable}{cor}{corone}
\label{corone}
\[
	\sqrt{n}( \tilde\Beta_n -\sigma_*\Sigma_n\inv\Sigma\Beta_0) \Rightarrow \N_p(0, \Sigma\inv A\Sigma\inv), 
\]
where $\Sigma_n\inv = (n\inv\sum_{i=1}^n \x_i'\x_i)\inv$.
\end{restatable}

\begin{restatable}{cor}{cortwo}
\label{cortwo}
If the covariance matrix, $\Sigma$, is known then
\[
	\sqrt{n}( \dot\Beta_n -\sigma_*\Beta_0) \Rightarrow \N_p(0, \Sigma\inv A\Sigma\inv), 
\]
where
\[
 \dot\Beta_n = n\inv\Sigma\inv\sum_{i=1}^n\x_i'\Phi_n\inv(H_n^*(Y_i))
\]
\end{restatable}

\section{Simulations}
Our Truncated Gaussian Quantile Regression estimate \eqref{eq:tgqr} was
pitted against the Spearmax estimate \eqref{eq:spearmax} and the usual ordinary least squares
(OLS) estimate gotten from the full data with raw response. The Spearmax estimate was achieved by parameterizing $\Beta \in \B$, where $\B$ is the unit circle,
by $\theta = \arctan(\beta_2/\beta_1)$. The solution to
\begin{equation}
  \argmax_{\theta\in(0,2\pi)} \sum_{i=1}^n R_n(Y_i)R_n (\x_i(\cos\theta, \sin\theta)')
\end{equation}
was achieved by bound constrained numerical optimization \citep{byrd_limited_1995}. The constrained region was centered on the true $\theta_0=\arctan(2/1) = 26.6\degree$. Figure~\ref{fig:spearmax} demonstrates the function to be maximized for a particular sample of size twenty under the Gaussian conditions of the first simulation scenario \eqref{eq:gausssim}.

We also consider an Empirical Quantile Regression estimate 
\begin{equation}\label{eq:eqr}
  \tilde\tilde \Beta_n = \left(n\inv \sum_{i=1}^n \x_i'\x_i \right)\inv
   n\inv \sum_{i=1}^n \x_i'\hat\Phi\inv_{\x_i'\Beta_n}(H_n^*(Y_i))
\end{equation}
where $\hat\Phi\inv_{\x_i'\Beta_n}$ is the inverse of the empirical distribution of $\x_i'\tilde\tilde\Beta_n$.

Covariates and errors were simulated under the Gaussian and stable distributions as described below. From each of the 10,000 simulated trials, the estimated angle of $\hat\Beta$, $\hat\theta$, was recorded using each of the methods. The sample standard deviation and bias from the 10,000 trials is graphically summarized for each scenario.

\begin{figure}
\begin{knitrout}
\definecolor{shadecolor}{rgb}{0.969, 0.969, 0.969}\color{fgcolor}
\includegraphics[width=\linewidth]{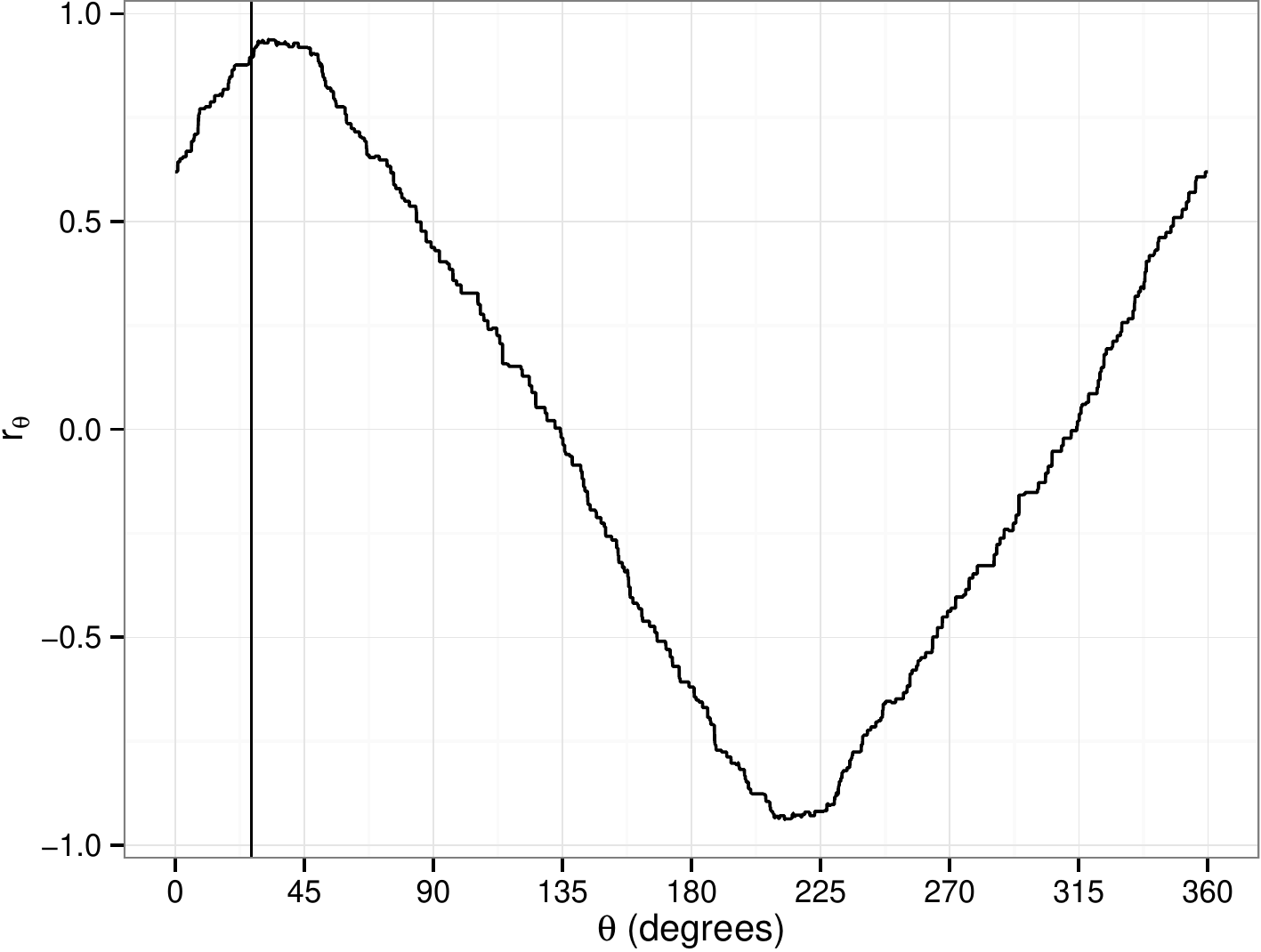} 

\end{knitrout}
\caption{Spearman's correlation between $Y$ and $\x'\Beta_\theta$, where $\Beta_\theta = (\cos(\theta), \sin(\theta))$ for a random sample of $n$ = 20. The vertical line marks true $\theta_0=\arctan(1/2)=26.6\degree$.}
\label{fig:spearmax}
\end{figure}

\subsection{Gaussian simulations}
The first run of simulations follow
\begin{equation}\label{eq:gausssim}
\x_i \sim 
N\left(
0,
\left(
\begin{array}{cc}
  1 & 0 \\
  0 & 2
\end{array}
\right)
\right), \quad
\Beta_0 = 
\left(
\begin{array}{c}
  2 \\
  1 
\end{array}
\right), \quad
\varepsilon_i \sim \N(0,1), i=1,\ldots,n. 
\end{equation}
Simulated sample sizes ranged between $n = 25$ to $n=3,000$. Results are summarized in Figure~\ref{fig:sim1}. Gaussian Quantile regression performs as well as OLS on the full data with moderate sample sizes. Spearmax performs best on small samples and the Empirical Quantile Regression performs best on moderate sample sizes.

\begin{figure}
\begin{knitrout}
\definecolor{shadecolor}{rgb}{0.969, 0.969, 0.969}\color{fgcolor}
\includegraphics[width=\linewidth]{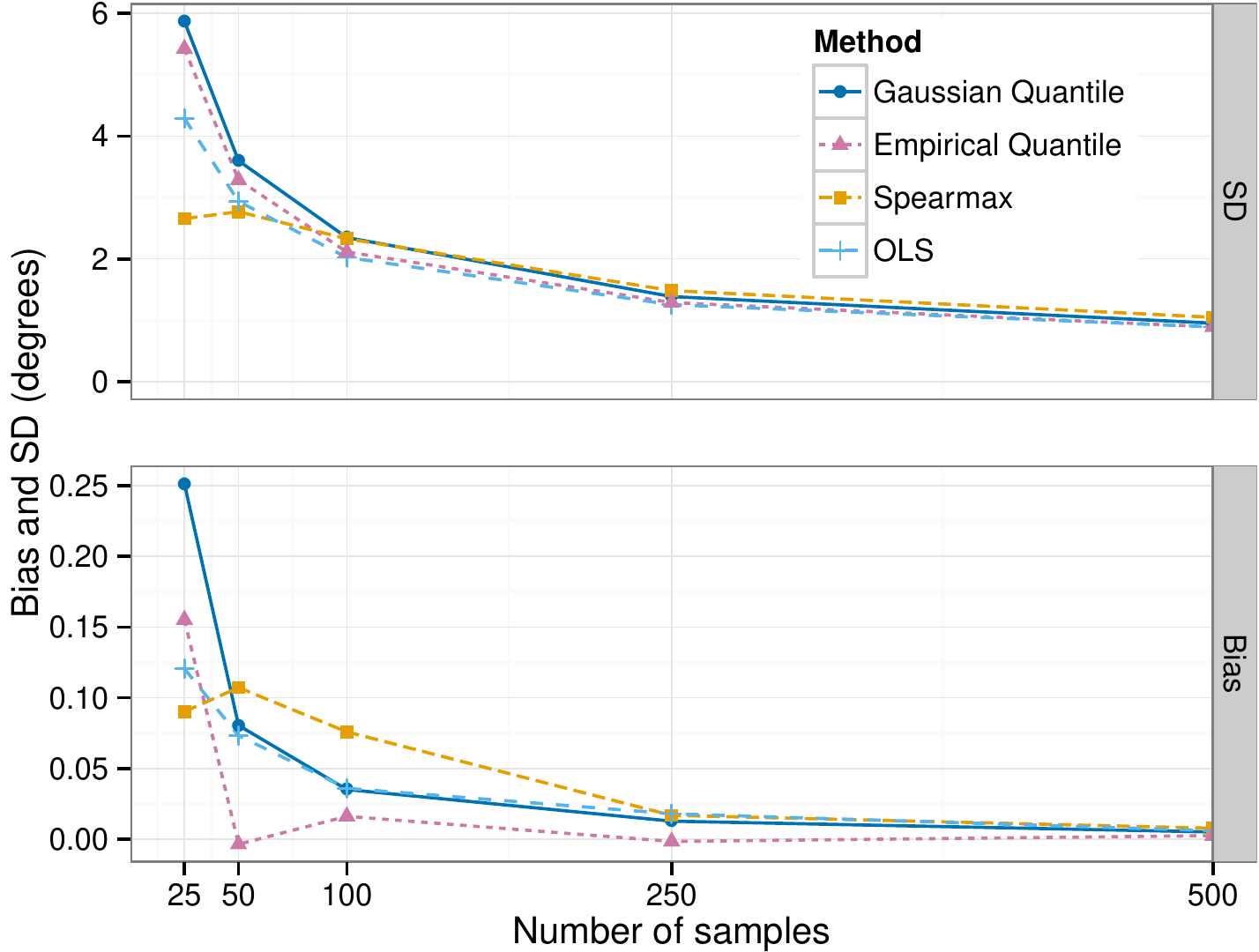} 

\end{knitrout}
\caption{Gaussian simulations. Simulated standard deviation and bias of the point estimates from 10,000 simulations of varying sample size. Covariates and error are simulated under Gaussian distributions.}
\label{fig:sim1}
\end{figure}

\subsection{Impact of skewness}

The second scenario simulated the two covariates and errors, independently from stable distributions with fixed stability parameter $\alpha=1$, sample size $n=500$ and the skewness parameter $\beta$ ranging in -1 to 1. Results of this scenario are summarized in Figure~\ref{fig:sim2}. Bias and SD of the Truncated Gaussian Quantile estimator is greater with this heavy tailed distribution ($\alpha = 1$), but it is unaffected by skewness parameter. The Empirical Quantile distribution appears to mitigate some of the bias and SD induced by the skewness of the covariate distribution.

\begin{figure}
\begin{knitrout}
\definecolor{shadecolor}{rgb}{0.969, 0.969, 0.969}\color{fgcolor}
\includegraphics[width=\linewidth]{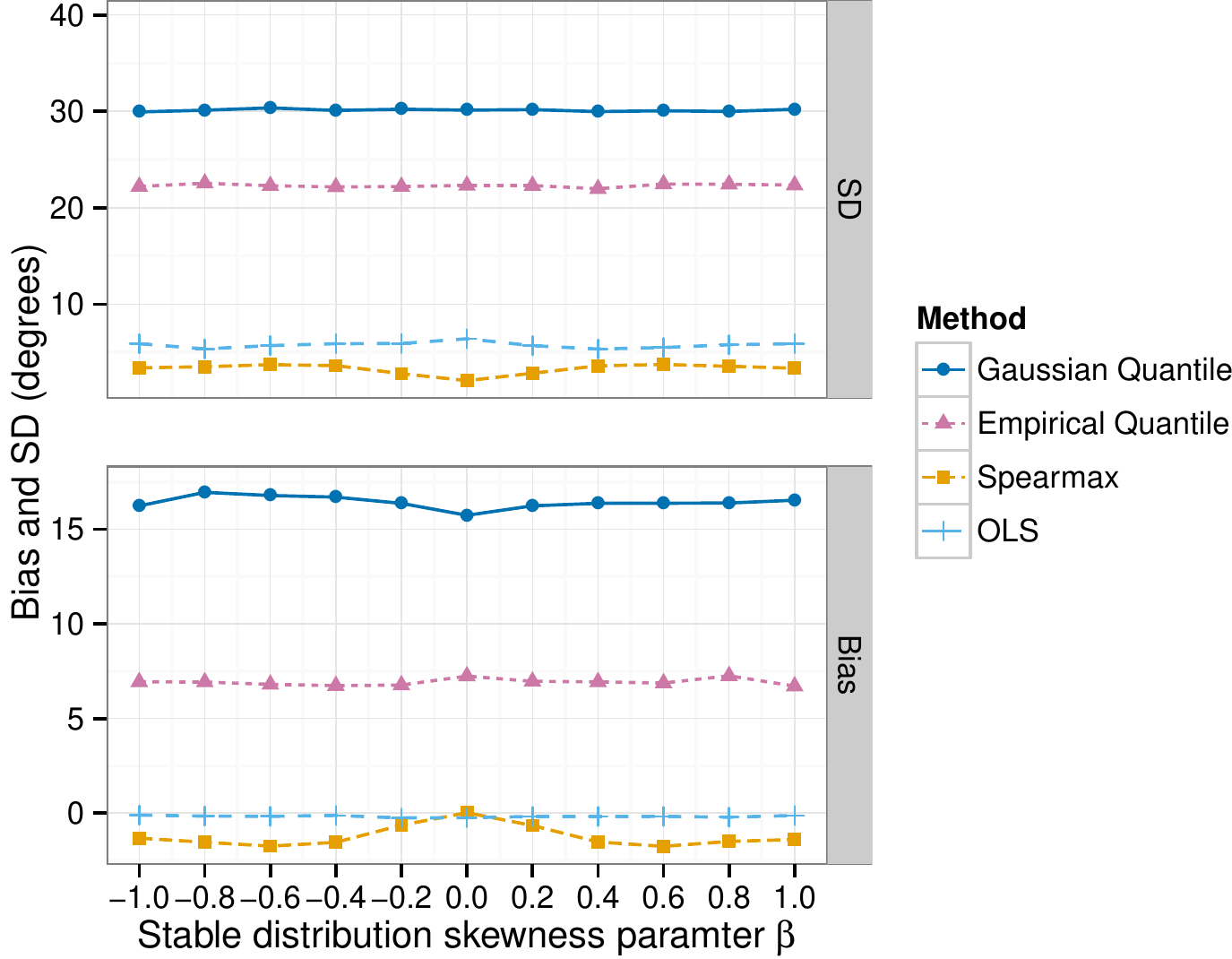} 

\end{knitrout}
\caption{Impact of skewness. Simulated standard deviation and bias of the three point estimates from 10,000 simulations of size $n=500$ drawing covariates and errors from the stable distribution with stability parameter $\alpha=1$ and the given skewness parameter $\beta$.}
\label{fig:sim2}
\end{figure}

\subsection{Impact of stability}

The third scenario simulated the two covariates and errors, independently from stable distributions with fixed skewness parameter $\beta=0$, sample size $n=500$, and
stability parameter $\alpha$ ranging in 0.2 to 2.0 (Gaussian). Results of this scenario are summarized in Figure~\ref{fig:sim3}. Here Spearmax shows an advantage over OLS, and again, the Empirical Quantile Regression seems to mitigate some of the bias and SD induced by the heavy tails of the distribution of $\x'\Beta_0$.

\begin{figure}
\begin{knitrout}
\definecolor{shadecolor}{rgb}{0.969, 0.969, 0.969}\color{fgcolor}
\includegraphics[width=\linewidth]{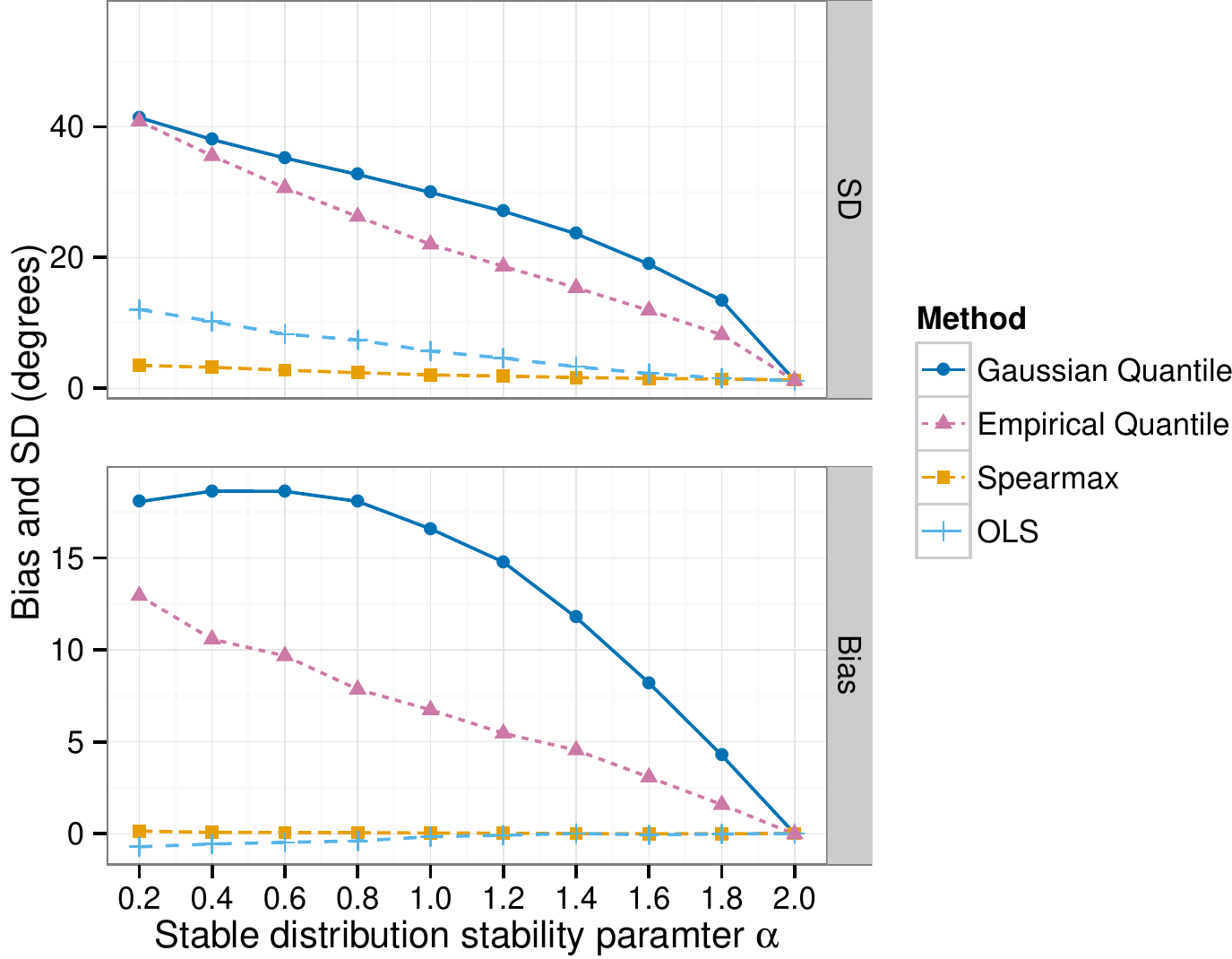} 

\end{knitrout}
\caption{Impact of stability. Simulated standard deviation and bias of the three point estimates from 10,000 simulations of size $n=500$ drawing covariates and errors from the stable distribution with skewness parameter $\beta=0$ and the given stability parameter $\alpha$.}
\label{fig:sim3}
\end{figure}

\section{Application to Alzheimer's blood plasma assay}

We apply the three rank-based regression models to baseline blood plasma assays of amyloid-$\beta_{1-42}$.
This data has been previously described and analyzed \citep{donohue_longitudinal_2014}. For the present
analysis, we focus on the association between response variable amyloid-$\beta_{1-42}$ and predictors age at baseline and Alzheimer's Disease Assessment Scale (ADAS) in the cohort with Mild Cognitive Impairment (MCI). Assays were performed in duplicate for each participant. The means of duplicate florescence intensities were rank transformed plate by plate, under the assumption that the distribution of the assay on each plate should be similar. The original investigation, using random effects to model duplicate observations and plate effects, found that amyloid-$\beta_{1-42}$ increased 0.16 pg/ml/year of age (SE=0.08, p=0.047), and in a separate model, found no significant association between ADAS and amyloid-$\beta_{1-42}$.

Table \ref{tab:alz} shows the results of the alternative rank-based analyses. The Truncated Gaussian Quantile regression method found that increased amyloid-$\beta_{1-42}$ was associated with better ADAS scores (ADAS increases with worsening); while the Empirical Quantile method found that increased amyloid-$\beta_{1-42}$ was associated increased age. The age association is consistent with the original investigation. The Spearmax method found both coefficients to be significant at the 0.05 level. Notably, there was strong evidence that the distribution of $\x\hat'\Beta$ from the Truncated Gaussian Quantile regression was not Gaussian (Anderson-Darling $p<0.001$). The R functions used for this analysis are provided in the Appendix.

\begin{table}[!tbp]
 \begin{center}
 \begin{tabular}{lrrrl}\hline\hline
\multicolumn{1}{l}{coefficient}&\multicolumn{1}{c}{estimate}&\multicolumn{1}{c}{lower}&\multicolumn{1}{c}{upper}&\multicolumn{1}{c}{}\tabularnewline\hline
{\bfseries Truncated Gaussian Quantile}&&&&\tabularnewline
~~Age&$ 0.1631$&$-0.114$&$ 0.4852$&\tabularnewline
~~ADAS&$-0.9866$&$-1.964$&$-0.2361$&*\tabularnewline
\hline
{\bfseries Empirical Quantile}&&&&\tabularnewline
~~Age&$ 0.9969$&$ 0.981$&$ 1.0233$&*\tabularnewline
~~ADAS&$ 0.0785$&$-0.121$&$ 0.2687$&\tabularnewline
\hline
{\bfseries Spearmax}&&&&\tabularnewline
~~Age&$ 0.7761$&$ 0.499$&$ 1.0278$&*\tabularnewline
~~ADAS&$-0.6306$&$-1.499$&$-0.0702$&*\tabularnewline
\hline
\end{tabular}

\end{center}
\caption{Estimates of the effect of age, Alzheimer's Disease Assessment Scale (ADAS), and their interaction on a blood plasma assay of amyloid-$\beta_{1-42}$ using three rank-based regression methods.
Bootstrap 95\% confidence limits are estimated by 1,000 bootstrap resamples. Confidence intervals that
exclude 0 are marked by ``*''.}
\label{tab:alz}
\end{table}

\section{Discussion}

We present a novel regression estimator applicable to
\cite{han_non-parametric_1987} monotonic index model under Gaussian assumptions,
when only the ranks of the responses and covariates are observed. The
proposed estimator is shown in simulations to be competitive with
the known alternative, Spearmax, when the covariates and errors are Gaussian. 
We prove the estimator is consistent and
asymptotically normal. The appendix also demonstrate how to attain consistent
bootstrap confidence regions in the bivariate case. An obvious
drawback of our proposed rank-based estimate, compared to the
\cite{sherman_limiting_1993} class of estimates, is the strong Gaussian
assumptions. The assumption, however, yields the key equality \eqref{eq:keyobs}, which allows $\Beta_0$ to emerge asymptotically from the composition, $\Phi\inv(H_n^*(Y_i))$. This in turn allows for an estimate with computational advantages over Spearmax.

The rank-based regression estimator demonstrated advantages
over the Spearmax estimate in terms of simulated standard error and
bias in moderately sized ($n\geq50$) fully Gaussian setting. Although Spearmax is more robust, when the
dimension of the target parameter is large the computational benefits of
the rank-based regression may make it an attractive alternative. Where the
Spearmax estimator requires maximizing a step function over
$p-$dimensions, the Gaussian quantile estimate requires only the
ubiquitous numerical approximation to the Gaussian quantile function and 
the ordinary least squares machinery. In the Alzheimer's example, the Truncated
Gaussian approach required about 50\% of the system time required by Spearmax.

Simulations demonstrate that Gaussian Quantile Regression is sensitive to heavy tailed covariates and errors, but fairly robust to skewness. Empirical Quantile Regression offers some mitigation of effects of heavy tails, but Spearmax is recommended if computationally feasible and there is reason to believe the errors might be heavy tailed.

\section{Appendix}

\subsection{Proof of Lemma~\ref{lemtwo}}
\lemtwo*
\begin{proof}
\begin{align*}
\sup_{Y_i} |H_n^*(Y_i) - H(Y_i)| 
& \leq \sup_{y} \left| \frac{n}{n+1} H_n(y) - H(y) \right|\\
& \leq \sup_{y} \left| \frac{n}{n+1} H_n(y) - H(y) \right| + 
  \sup_{y}|H_n(y) - H(y)|\\
& \leq 1 - \frac{n}{n+1} + \sup_{y}|H_n(y) - H(y)| \tag*{$H_n\leq1$}\\
& \as 0. \tag*{Glivenko-Cantelli}
\end{align*}
Therefore we get
\begin{align*}
\xi_{nij} & = X_{ij}\Phi\inv (H_n^* (Y_i))\\
& \as X_{ij}\Phi\inv (H(Y_i))\\
& = \frac{X_{ij}(\x_i'\Beta_0 + \varepsilon_i)}{\sqrt{\sigma_{\Beta_0}^2 + \sigma^2}}\\
& = \xi_{ij}.
\end{align*}
\end{proof}

\subsection{Proof of Lemma~\ref{lemthree}}
\lemthree*
\begin{proof}
This is an expression of Renyi's Theorem \cite[p.~96]{resnick_probability_1999} which states
\[
\p[R_n(Y_i) = k] = n\inv\textrm{ for }k \in\{1, \dots, n\}
\]
\end{proof}

\subsection{Proof of Lemma~\ref{lemfour}}
\lemfour*
\begin{proof}
For $n\geq 1$, note that $2n \geq n + 1$ and therefore $n\inv  \leq 2(n + 1)\inv$.
So we have: 
\begin{align*}
\E\Phi\inv(U_n)^4 & = \sum_{k=1}^n\Phi\inv (k(n +1)\inv )^4 n\inv \\
& \leq \sum_{k=1}^n\Phi\inv (k(n +1)\inv )^4 2(n+1)\inv \\
& = \frac{2}{n+1}\left(
  \sum_{k=1}^{\floor{\frac{n}{2}}}\Phi\inv (k(n +1)\inv )^4 +
  \sum^{n}_{k=\ceil{\frac{n}{2}}}\Phi\inv (k(n +1)\inv )^4\right)\\
& \leq 2\int_{1}^{1/2}\Phi\inv (u)^4 du +
       2\int^{1}_{1/2}\Phi\inv (u)^4 du\\
& = 2\E\Phi\inv(U)^4\\
& = 2\E Z^4 = 6\\
\end{align*}
The third line in the above follows because $\Phi\inv(1/2)^4= 0$. This
allows the Riemann approximation to be left-handed where
$\Phi\inv(\cdot)^4$ is decreasing $(k < n/2)$, and right-handed where
$\Phi\inv(\cdot)^4$ is increasing $(k > n/2)$. Figure \ref{fig:reiman} (with $n
= 20$) demonstrates that the approximation can always be kept below the
curve $\Phi\inv(\cdot)^4$.
\end{proof}

\begin{figure}
\begin{knitrout}
\definecolor{shadecolor}{rgb}{0.969, 0.969, 0.969}\color{fgcolor}
\includegraphics[width=\linewidth]{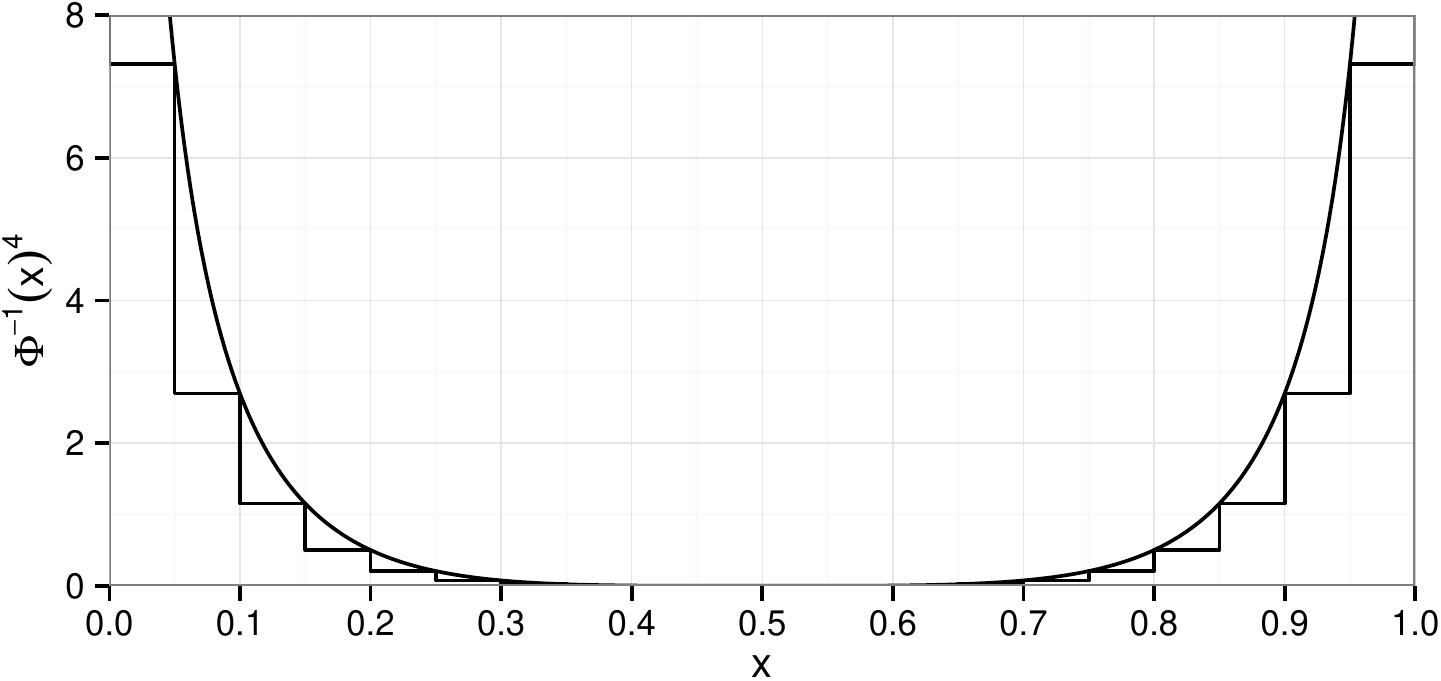} 

\end{knitrout}
\caption{Depiction of the Riemann approximation of $\Phi\inv(x)^4$ used in Lemma \ref{lemfour}}
\label{fig:reiman}
\end{figure}

\subsection{Proof of Lemma~\ref{lemfive}}
\lemfive*
\begin{proof}
\begin{align*}
\E\xi_{nij}^2 &= \E[X_{ij} \Phi\inv(H_n^*(Y_i))]^2\\
& \leq \sqrt{\E X_{ij}^4\E\Phi\inv(H_n^*(Y_i))^4} \tag*{Schwarz Inequality}\\
& = \sqrt{\E X_{ij}^4\E\Phi\inv(U_n)^4} \tag*{Lemma \ref{lemthree}}\\
& = \sqrt{3\sigma_j^4\cdot6} = 3\sqrt{2}\sigma_j^2 \tag*{Lemma \ref{lemfour}}
\end{align*}
\end{proof}

\subsection{Proof of Lemma~\ref{lemsix}}
\lemsix*
\begin{proof}
First, $\{\xi_{nij}\}$ is uniformly integrable via the Crystal Ball Condition
\cite[p.~184]{resnick_probability_1999} since 
$\sup_n\E|\xi_{nij}|\leq 3\sqrt{2}\sigma^2_j$ 
by Lemma \ref{lemfive}. Uniform integrability and $\xi_{nij} \as \xi_{ij}$ gives us 
$\E\xi{nij}\rightarrow \E\xi_{ij}$ \cite[p.~191]{resnick_probability_1999}
\end{proof}

\subsection{Proof of Lemma~\ref{lemseven}}
\lemseven*
\begin{proof}
Now $\{\xi_{nsj}\xi_{ntj}\}$ is also u.i. since
\[
\E|\xi_{nsj}\xi_{ntj}| \leq \sqrt{\E\xi_{nsj}^2\E\xi_{ntj}^2} \leq 3\sqrt{2}\sigma_j^2 < \infty
\]
We also have $\xi_{nsj}\xi_{ntj} \as \xi_{sj}\xi_{tj}$ and so 
$\E[\xi_{nsj}\xi_{ntj}] \rightarrow \E[\xi_{sj}\xi_{tj}]$. 
\begin{align*}
\cov(\xi_{nsj},\xi_{ntj}) & = \E[\xi_{nsj}\xi_{ntj}] - \E\xi_{nsj}\E\xi_{ntj}\\
& \rightarrow \E[\xi_{sj}\xi_{tj}] - \E\xi_{sj}\E\xi_{tj} = 0, 
\end{align*}
where the last equality follows by independence.
\end{proof}

\subsection{Proof of Lemma~\ref{lemeight}}
\lemeight*
\begin{proof}
For any $\delta>0$ we have by the Markov Inequality
\begin{align*}
  \p[\|&S^{(2)}_n - S^{(1)}_n\| > \delta] \leq \delta^{-2}\E\|S^{(2)}_n - S^{(1)}_n\|^2\\
  & = \delta^{-2}n\inv\E\left\| \sum_{i=1}^n \x_i' (\Phi\inv (H(Y_i)) - \Phi\inv_n(H(Y_i)))\right\|\\
  & \leq \delta^{-2}n\inv\E\sum_{i=1}^n\left\| \x_i' (\Phi\inv (H(Y_i)) - \Phi\inv_n(H(Y_i)))\right\|\\
  & = \delta^{-2}\E\left\| \x_i'(\Phi\inv (H(Y_i)) - \Phi\inv_n(H(Y_i)))\right\|\\
  & \leq \delta^{-2}\E\left[\|\x_i\|^2 \Phi\inv (H(Y_i))^2\{H(Y_i) \in (0,1-\alpha) \cup (\alpha_n,1)\}\right]\\
  & = \delta^{-2}\sigma_*^2\E\left[\|\x_i(\x_i'\Beta_0+\varepsilon_i)\|^2 \{H(Y_i) \in (0,1-\alpha) \cup (\alpha_n,1)\}\right]\\
  & \rightarrow 0. \tag{DCT}
\end{align*}
The last line follows from the Dominated Convergence Theorem since
\[
\|\x_i(\x_i'\Beta_0 + \varepsilon_i)\|^2\{H(Y_1) \in (0, 1 - \alpha_n ) \cup (\alpha_n , 1)\}
\as 0
\]
and
\[
\|\x_i(\x_i'\Beta_0 + \varepsilon_i)\|^2\{H(Y_1) \in (0, 1 - \alpha_n ) \cup (\alpha_n , 1)\} \leq \|\x_i(\x_i'\Beta_0+\varepsilon_i)\|^2\in L_1
\]
by the Gaussian assumptions on $\x_i$.
\end{proof}

\subsection{Proof of Lemma~\ref{lemnine}}
\lemnine*
\begin{proof}
Note that the slope of $\Phi\inv$ at $x$ is the reciprocal of the slope of $\Phi$
at $\Phi\inv(x)$, which can be evaluated in terms of the density as $\phi(\Phi\inv(x))$. That is, the
first derivative of $\Phi\inv$ for $x\in(0, 1)$ is
\begin{align*}
 \frac{d}{dx} \Phi\inv(x) & = \frac{1}{\phi(\Phi\inv(x))} \\
 & = \sqrt{2\pi}\exp(\Phi\inv(x)^2/2). 
\end{align*}
Recall $\Phi\inv$ is symmetric about $1/2$ and $\Phi\inv_n$ 
preserves that symmetry. Also,
the slope of $\Phi\inv_n$ is maximized at the truncation points, i.e.
\begin{align*}
 \frac{d}{dx} \Phi\inv_n(x) & < \frac{1}{\phi(\Phi\inv(\alpha_n))} \\
 & = \sqrt{2\pi}\exp\left(\frac{1}{2}\Phi\inv\left(\Phi\left(\sqrt{\frac{1}{2}\log n}\right)\right)^2\right)\\
 & = \sqrt{2\pi}n^{1/4}.
\end{align*}
By a first order Taylor series approximation we have 
\begin{align*}
(\Phi\inv_n(H(Y_1)) - \Phi\inv_n(H_n^*(Y_1)))^2 
 & \leq (\sqrt{2\pi}n^{1/4}(H(Y_1) - H_n^*(Y_1)))^2\\
 & = 2\pi\sqrt{n}(H(Y_1)-H_n^*(Y_1))^2
\end{align*}
as desired.
\end{proof}

\subsection{Proof of Lemma~\ref{lemten}}
\lemten*
\begin{proof}
We have 
\begin{align*}
	\E_1 & (H(Y_1) - H_n^*(Y_1))^2 \\
= & H(Y_1)^2 - 2H(Y_1)\E_1 H_n^*(Y_1) + \E_1 H_n^*(Y_1)^2\\
= & H(Y_1)^2 - 2H(Y_1)\left(\frac{1}{n+1} + \frac{n-1}{n+1}H(Y_1)\right)\\
  & + \E_1\left(\frac{1}{n+1} + \frac{1}{n+1}\sum_{i=2}^n\{Y_i\leq Y_1\}\right)^2\\
= & \frac{3-n}{n+1}H(Y_1)^2 - \frac{2}{n+1}H(Y_1) + \frac{1}{(n+1)^2}\\
  & + \E_1\left(\frac{2}{(n+1)^2}\sum_{i=2}^n\{Y_i\leq Y_1\} + \frac{1}{(n+1)^2}\left(\sum_{i=2}^n\{Y_i\leq Y_1\}\right)^2\right)\\
= & \frac{3-n}{(n+1)^2}H(Y_1)^2 - \frac{4}{(n+1)^2}H(Y_1) + \frac{1}{(n+1)^2}\\
  & + \frac{1}{(n+1)^2}\left((n-1)H(Y_1)+(n-1)(n-2)H(Y_1)^2\right)\\
= & \frac{5-n}{(n+1)^2}H(Y_1)^2 - \frac{n-5}{(n+1)^2}H(Y_1) + \frac{1}{(n+1)^2}\\
\leq & \frac{1}{n+1}	
\end{align*}
almost surely. The last inequality follows from
$H(Y_1 )\in [0, 1]$ almost surely. 
\end{proof}

\subsection{Proof of Lemma~\ref{lemeleven}}
\lemeleven*
\begin{proof}
\begin{align*}
\p[\| & S^{(1)} - S_n\| > \delta] \leq \delta^{-2}\E\| S^{(1)}_n - S_n\|^2\\
& = \delta^-2n\inv\E\left\|\sum_{i=1}^n \x_i(\Phi\inv_n(H(Y_i)) - \Phi\inv_n(H_n^*(Y_i)))\right\|^2\\
& \leq \delta^-2n\inv\E\sum_{i=1}^n\left\| \x_i(\Phi\inv_n(H(Y_i)) - \Phi\inv_n(H_n^*(Y_i)))\right\|^2\\
& = \delta^-2\E\left\| \x_i(\Phi\inv(H(Y_i)) - \Phi\inv_n(H(Y_i)))\right\|^2\\
& \leq \delta^-2\sigma_*^2 2\pi\sqrt{n}\E\left[\| \x_i \|^2(H(Y_i) - H_n^*(Y_i))^2\right]\\
& =    \delta^-2\sigma_*^2 2\pi\sqrt{n}\E\left[\| \x_i \|^2 \E_1(H(Y_i) - H_n^*(Y_i))^2\right]\\
& \leq \frac{\sigma_*^2 2\pi\sqrt{n}}{\delta^2(n+1)}\E\|\x_i\|^2\\
& \rightarrow 0.
\end{align*}
\end{proof}
	
\subsection{Proof of Lemma~\ref{corone}}
\corone*
\begin{proof}
Since $\Sigma\inv_n\ip\Sigma$, we have
\[
\sqrt{n}( \tilde\Beta_n -\sigma_*\Sigma_n\inv\Sigma\Beta_0) = \Sigma_n\inv S_n\Rightarrow\N_p(0,\Sigma\inv A\Sigma\inv)
\]
\end{proof}

\subsection{Proof of Lemma~\ref{cortwo}}
\cortwo*
\begin{proof}
Since $\Sigma\inv_n\ip\Sigma$, we have
\[
\sqrt{n}( \dot\Beta_n -\sigma_*\Beta_0) = \Sigma_n\inv S_n\Rightarrow\N_p(0,\Sigma\inv A\Sigma\inv)
\]
\end{proof}

\subsection{Confidence intervals}
Based on the
asymptotic normality of our estimator, we can produce confidence regions
by estimating its bias and dispersion matrix. Alternatively we can
use the weak dependence property that accompanies the established
$\sqrt{n}$-consistency to construct jackknife confidence intervals. This
amounts to a bootstrap for triangular arrays which are identically
distributed across rows with a weak dependence condition. The theory
is similar in spirit to the stationary time series bootstrap, except
the dependence here is not serial. Recall the setup is $Y = X\Beta_0 + \varepsilon$
along with our Gaussian assumptions and that we observe
\[
\X = \left(\begin{array}{c}
\x_1\\
\x_2\\
\vdots\\
\x_n\\
\end{array}
\right)\textrm{ and } \R_n(\Y) = \left(\begin{array}{c}
R_n(Y_1)\\
R_n(Y_2)\\
\vdots\\
R_n(Y_3)\\
\end{array}
\right).
\]

Our asymptotically normal estimator is 
\[
\tilde\Beta_n = \left(n\inv \sum_{i=1}^n \x_i'\x_i\right)\inv n\inv\sum_{i=1}^n \x_i'\Phi_n\inv(H_n^*(Y_i)).
\]
and we have shown $\sqrt{n}( \tilde\Beta_n -\sigma_*\Sigma_n\inv\Sigma\Beta_0) = \Sigma_n\inv S_n\Rightarrow\N_p(0,\Sigma\inv A\Sigma\inv)$. Suppose $p=2$ and let
\[
\theta_{0n} = \arctan\left(\frac{\e_2\Sigma_n\inv\Sigma\Beta_0}{\e_1\Sigma_n\inv\Sigma\Beta_0}\right)\ip\arctan(\Beta_{02}/\Beta_{01})
\]
and
\[
\tilde\theta_n = \arctan(\tilde\Beta_{n2}/\tilde\Beta_{n1})
\]
By the bivariate delta method \cite{lehmann_elements_1999} we have
\begin{equation}\label{eq:17}
\sqrt{n}(\tilde\theta_n - \theta_{0n}) \Rightarrow \N 
\end{equation}
with zero mean provide $\Beta_{n1}$ is bounded away from zero.

In the previous section we explored the
covariance structure \eqref{eq:15} of our estimate, namely the limiting
dispersion matrix $A$. Since $\sqrt{n}(\tilde\theta_n - \theta_{0n})$ is
admittedly not a pivot, we are forced to estimate its variance if
confidence intervals using the normal approximation are desired.
Adapting the standard methodology from the stationary bootstrap
literature \cite{kunsch_jackknife_1989, lahiri_second_1991, politis_general_1992}, 
as well as the order statistics literature \cite{sen_12_1998}, we
describe an appropriate functional jackknife procedure.

Observe the
$n$ leave-one-out jackknife subsamples of the form
\begin{align*}
	\tilde\Beta_{-i}^* & = (n-1)\inv\left(\sum_{\{j:j\neq i\}} \x_j'x_j\right) \sum_{\{j:j\neq i\}} \x_j'\Phi_{(n-1)}\inv(H_{-i}^*(Y_j)),\\
	\tilde\theta_{-i}^* & = \arctan(\tilde\Beta_{-i2}^*/\tilde\Beta_{-i1}^*), 
\end{align*}	
where $H_{-i}^*(Y_j) = n\inv\sum_{\{k:k\neq i\}}\{Y_k\leq Y_j\}$.
Now our estimate of the variance is the
sample variance of the subsample estimates
\begin{equation}\label{eq:18}
\var(\tilde\theta_n) \approx \tilde\sigma_n^2 = n\inv\sum_{i=1}^n\tilde\theta_{-i}^{*2}-\left(n\inv\sum_{i=1}^n\tilde\theta_{-i}^*\right)^2.
\end{equation}

If the bias is negligible, we have an approximate $(1 - \alpha)100\%$ confidence
interval:
\[
 \left[\tilde\theta_n - z_\alpha\tilde\sigma_n, \tilde\theta_n + z_\alpha\tilde\sigma_n\right], 
\]
where $z_\alpha=\Phi\inv(1 - \alpha/2)$. We can also include a bias correction by estimating the
mean of our estimate: 
\[
\E(\tilde\theta_n)\approx \tilde\mu_n = n\inv\sum_{i=1}^n\tilde\theta_{-i}^*.
\]
The bias corrected approximate confidence interval then takes the form
\[
 \left[\tilde\theta_n - \tilde\mu_n - z_\alpha\tilde\sigma_n, \tilde\theta_n - \tilde\mu_n + z_\alpha\tilde\sigma_n\right]. 
\]

To justify the consistency of these estimates consider, by Chebychev's Inequality and \eqref{eq:17}
\[
\p[|\tilde\mu_n-\E(\tilde\mu_n)|>\delta]\leq \delta^{-2}\var(\tilde\mu_n) \leq \delta^{-2}\var(\tilde\theta_{-1}^*)=\delta^{-2}\var(\tilde\theta_{n-1})\rightarrow 0.
\]
We also have
\[
|\E(\tilde\theta_n) - \E(\tilde\mu_n)| = |\E(\tilde\theta_n) - \E(\tilde\theta_1^*)| = |\E(\tilde\theta_n) - \E(\tilde\theta_{n-1})|\rightarrow 0 
\]
so that 
\[
|\E(\tilde\theta_n) - \tilde\mu_n| \leq |\E(\tilde\theta_n) - \E(\tilde\mu_n)| + |\E(\tilde\mu_n) - \tilde\mu_n| \ip 0. 
\]
Therefore, $\tilde\mu_n$ is a consistent estimate for $\E(\tilde\theta_n)$. 

By another application of the delta method we have $\sqrt{n}(\tilde\theta_n^2-\theta_{0n}^2)\Rightarrow\N$
with mean zero. We then have
\begin{align*}
\p\left[\left| n\inv\sum_{i=1}^n\tilde\theta_{-i}^{*2}-\E(\tilde\theta_{1}^{*2})\right|>\delta\right]
 & \leq \delta^{-2}\var\left(  n\inv\sum_{i=1}^n\tilde\theta_{-i}^{*2} \right)\\
 & \leq \delta^{-2}\var(\tilde\theta_{-1}^{*2})\\
 & \leq \delta^{-2}\var(\tilde\theta_{n-1}^{2})\\
 & \rightarrow 0.
\end{align*}

Additionally, we have
\[
|\E(\tilde\theta_n^2) - \E(\tilde\theta_{-1}^{*2})| = 
|\E(\tilde\theta_n^2) - \E(\tilde\theta_{n-1}^{2})| \rightarrow 0,
\]
so that
\[
\left| \E(\tilde\theta_{n}^2) - n\inv\sum{i=1}^n \tilde\theta_{-i}^{*2} \right|
 \leq \left| \E(\tilde\theta_{n}^2) - \E(\tilde\theta_{-1}^{*2}) \right| +
      \left| \E(\tilde\theta_{-1}^{*2}) - n\inv\sum{i=1}^n \tilde\theta_{-i}^{*2} \right|
 \rightarrow 0.
\]

Finally we have consistency of our variance estimate, $|\tilde\sigma_n^2 - \var(\tilde\theta_n)| \ip 0$.

Recall that when the dispersion, $\Sigma$, is known, Corollary \ref{cortwo} gives us asymptotic normality of
the form
\[
 \sqrt{n}(\dot\Beta_n - \sigma_*\inv\Beta_0) \Rightarrow \N_p(0,\Sigma\inv A\Sigma\inv), 
\] 
where
\[
\dot\Beta_n = n\inv \sum_{i=1}^n \x_i'\Phi\inv(H_n^*(Y_i)). 
\]
In this case, with $p = 2$ we compute statistics of the form
\[
 \dot\theta_n = \arctan(\dot\Beta_{n2}/\Beta_{n1}) 
\]
such that,
\[ 
\sqrt{n}(\dot\theta_n - \theta_0) \Rightarrow \N, \theta_0 = \arctan(\Beta_{02}/\Beta{01}).
\]

We can derive consistent estimators of the mean and variance similar to the unknown dispersion
case based on the jackknife observations
\begin{align*}
	\Beta_{-i}^* & = (n-1)\inv\Sigma\inv\sum_{\{j:j\neq i\}} \x_j'\Phi\inv_{(n-1)}(H_n^*(Y_j)),\\
	\dot\theta_{-i}^* & = \arctan(\dot\Beta_{-i2}/\dot\Beta{-i1}).
\end{align*}
Additionally, with known $\Sigma$ we can estimate the distribution of $\dot\theta_n$ directly by
subsampling and using the quantiles of this approximation to form confidence intervals, rather than the normal
quantile approximations. Specifically we have
\[
\Dist_{\dot\theta_n}(x)\approx \Dist_{\dot\theta_n}^*(x) = n\inv \sum_{i=1}^n\{\dot\theta_{-i}^* \leq x\}.
\]
Let $\dot\theta_{(1)}^* \leq \dot\theta_{(2)}^* \leq \cdots \leq \dot\theta_{(n)}^*$, denote
the order statistics. then, 
\[
Q^*(\alpha/2)\approx\dot\theta_{(k_1)}^*, Q^*(1-\alpha/2)\approx\dot\theta_{(k_2)}^*,
\]
where $k_1 = \floor{B\alpha/2}, k_2 = \floor{B(1-\alpha/2)} + 1$. We get a percentile jackknife confidence interval for
$\Beta_0$ of the form
\[
[\dot\theta_{(k_1)}^*, \dot\theta_{(k_2)}^*].
\]
The confidence interval can be shown to be consistent by an argument similar
to that for the mean and variance estimators. 

We can further improve on the percentile jackknife confidence intervals by considering a 
studentized bootstrap-$t$ confidence interval. Consider
\[
S_\theta = \frac{\dot\theta_n - \theta_0}{\sqrt{\var(\dot\theta_n)}}.
\]
We can estimate $\Dist_{S_\theta}(x) = \p(S_\theta \leq x)$ by a nested jackknife
\[
\Dist_{S_\theta}^*(x) = n\inv\sum{i=1}^n\{ \dot\theta_i^* \leq x\sqrt{\hat\var(\dot\theta_i^*)}+\dot\theta_n\},
\]
where $\hat\var(\dot\theta_i^*)$ is estimated in the same manner as \eqref{eq:18}. The
estimate $\hat\var(\dot\theta_i^*)$ is gotten from the sample variance of $(n - 1)$
jackknife subsamples of size $(n - 2)$.

\subsection{R Functions}
\begin{knitrout}
\definecolor{shadecolor}{rgb}{0.969, 0.969, 0.969}\color{fgcolor}\begin{kframe}
\begin{alltt}
\hlcom{# General Truncated Gaussian Quantile function}
\hlstd{tgq} \hlkwb{<-} \hlkwa{function}\hlstd{(}\hlkwc{x}\hlstd{,} \hlkwc{cut}\hlstd{)}
\hlstd{\{}
  \hlstd{alpha} \hlkwb{<-} \hlkwd{pnorm}\hlstd{(cut)}
  \hlkwd{ifelse}\hlstd{(x} \hlopt{<} \hlnum{1}\hlopt{-}\hlstd{alpha,} \hlopt{-}\hlstd{cut,}
  \hlkwd{ifelse}\hlstd{(x} \hlopt{>} \hlstd{alpha, cut,} \hlkwd{qnorm}\hlstd{(x)))}
\hlstd{\}}

\hlcom{# H (rank) transformation}
\hlstd{fH} \hlkwb{<-} \hlkwa{function}\hlstd{(}\hlkwc{x}\hlstd{)}
\hlstd{\{}
  \hlkwd{rank}\hlstd{(x,} \hlkwc{na.last}\hlstd{=}\hlstr{'keep'}\hlstd{)}\hlopt{/}\hlstd{(}\hlkwd{length}\hlstd{(x[}\hlopt{!}\hlkwd{is.na}\hlstd{(x)])}\hlopt{+}\hlnum{1}\hlstd{)}
\hlstd{\}}

\hlcom{# Truncated Gaussian Quantile utility function for regression}
\hlstd{ftgq} \hlkwb{<-} \hlkwa{function}\hlstd{(}\hlkwc{x}\hlstd{)\{}
  \hlstd{n} \hlkwb{<-} \hlkwd{length}\hlstd{(x[}\hlopt{!}\hlkwd{is.na}\hlstd{(x)])}
  \hlkwd{tgq}\hlstd{(}\hlkwd{rank}\hlstd{(x,} \hlkwc{na.last}\hlstd{=}\hlstr{'keep'}\hlstd{)}\hlopt{/}\hlstd{(n}\hlopt{+}\hlnum{1}\hlstd{),} \hlkwc{cut}\hlstd{=}\hlkwd{sqrt}\hlstd{(}\hlkwd{log}\hlstd{(n)}\hlopt{/}\hlnum{2}\hlstd{))}
\hlstd{\}}

\hlstd{Length} \hlkwb{<-} \hlkwa{function}\hlstd{(}\hlkwc{x}\hlstd{)} \hlkwd{sqrt}\hlstd{(x}\hlopt{%*%}\hlstd{x)[}\hlnum{1}\hlstd{]}
\hlstd{Norm} \hlkwb{<-} \hlkwa{function}\hlstd{(}\hlkwc{x}\hlstd{) x}\hlopt{/}\hlkwd{Length}\hlstd{(x)}

\hlcom{# Truncated Gaussian Quantile Regression}
\hlstd{tgqr} \hlkwb{<-} \hlkwa{function} \hlstd{(}\hlkwc{formula}\hlstd{,} \hlkwc{data}\hlstd{,} \hlkwc{subset}\hlstd{,} \hlkwc{FUN} \hlstd{= ftgq,} \hlkwc{...}\hlstd{)}
\hlstd{\{}
  \hlstd{cl} \hlkwb{<-} \hlkwd{match.call}\hlstd{()}
  \hlstd{mf} \hlkwb{<-} \hlkwd{match.call}\hlstd{(}\hlkwc{expand.dots} \hlstd{=} \hlnum{FALSE}\hlstd{)}
  \hlstd{m} \hlkwb{<-} \hlkwd{match}\hlstd{(}\hlkwd{c}\hlstd{(}\hlstr{"formula"}\hlstd{,} \hlstr{"data"}\hlstd{,} \hlstr{"subset"}\hlstd{,} \hlstr{"weights"}\hlstd{,} \hlstr{"na.action"}\hlstd{,}
      \hlstr{"offset"}\hlstd{),} \hlkwd{names}\hlstd{(mf),} \hlnum{0L}\hlstd{)}
  \hlstd{mf} \hlkwb{<-} \hlstd{mf[}\hlkwd{c}\hlstd{(}\hlnum{1L}\hlstd{, m)]}
  \hlstd{mf}\hlopt{$}\hlstd{drop.unused.levels} \hlkwb{<-} \hlnum{TRUE}
  \hlstd{mf[[}\hlnum{1L}\hlstd{]]} \hlkwb{<-} \hlkwd{quote}\hlstd{(stats}\hlopt{::}\hlstd{model.frame)}
  \hlstd{mf} \hlkwb{<-} \hlkwd{eval}\hlstd{(mf,} \hlkwd{parent.frame}\hlstd{())}
  \hlstd{mf[,}\hlnum{1}\hlstd{]} \hlkwb{<-} \hlkwd{FUN}\hlstd{(}\hlkwd{model.response}\hlstd{(mf,} \hlstr{"numeric"}\hlstd{))}
  \hlstd{b} \hlkwb{<-} \hlkwd{lm}\hlstd{(formula,} \hlkwc{data}\hlstd{=mf, ...)}\hlopt{$}\hlstd{coef}
  \hlkwd{Norm}\hlstd{(b)}
\hlstd{\}}

\hlcom{# Empirical Quantile Regression}
\hlstd{eqr} \hlkwb{<-} \hlkwa{function} \hlstd{(}\hlkwc{formula}\hlstd{,} \hlkwc{data}\hlstd{,} \hlkwc{subset}\hlstd{,} \hlkwc{tol}\hlstd{=}\hlnum{1e-5}\hlstd{,}
  \hlkwc{maxiter}\hlstd{=}\hlnum{100}\hlstd{,} \hlkwc{truncate} \hlstd{=} \hlnum{FALSE}\hlstd{,} \hlkwc{...}\hlstd{)}
\hlstd{\{}
  \hlstd{cl} \hlkwb{<-} \hlkwd{match.call}\hlstd{()}
  \hlstd{mf} \hlkwb{<-} \hlkwd{match.call}\hlstd{(}\hlkwc{expand.dots} \hlstd{=} \hlnum{FALSE}\hlstd{)}
  \hlstd{m} \hlkwb{<-} \hlkwd{match}\hlstd{(}\hlkwd{c}\hlstd{(}\hlstr{"formula"}\hlstd{,} \hlstr{"data"}\hlstd{,} \hlstr{"subset"}\hlstd{,} \hlstr{"weights"}\hlstd{,} \hlstr{"na.action"}\hlstd{,}
      \hlstr{"offset"}\hlstd{),} \hlkwd{names}\hlstd{(mf),} \hlnum{0L}\hlstd{)}
  \hlstd{mf} \hlkwb{<-} \hlstd{mf[}\hlkwd{c}\hlstd{(}\hlnum{1L}\hlstd{, m)]}
  \hlstd{mf}\hlopt{$}\hlstd{drop.unused.levels} \hlkwb{<-} \hlnum{TRUE}
  \hlstd{mf[[}\hlnum{1L}\hlstd{]]} \hlkwb{<-} \hlkwd{quote}\hlstd{(stats}\hlopt{::}\hlstd{model.frame)}
  \hlstd{mf} \hlkwb{<-} \hlkwd{eval}\hlstd{(mf,} \hlkwd{parent.frame}\hlstd{())}
  \hlstd{mm} \hlkwb{<-} \hlkwd{model.matrix}\hlstd{(formula, mf)}
  \hlstd{R} \hlkwb{<-} \hlkwd{rank}\hlstd{(}\hlkwd{model.response}\hlstd{(mf,} \hlstr{"numeric"}\hlstd{),} \hlkwc{na.last}\hlstd{=}\hlstr{'keep'}\hlstd{)}
  \hlcom{#initial value}
  \hlstd{n} \hlkwb{<-} \hlkwd{nrow}\hlstd{(mm)}
  \hlstd{mf[,}\hlnum{1}\hlstd{]} \hlkwb{<-} \hlstd{R}\hlopt{/}\hlstd{(n}\hlopt{+}\hlnum{1}\hlstd{)}
  \hlstd{beta_eq} \hlkwb{<-} \hlkwd{Norm}\hlstd{(}\hlkwd{lm}\hlstd{(formula,} \hlkwc{data}\hlstd{=mf, ...)}\hlopt{$}\hlstd{coef)}
  \hlstd{p} \hlkwb{<-} \hlkwd{length}\hlstd{(beta_eq)}
  \hlstd{beta_eq0} \hlkwb{<-} \hlkwd{Norm}\hlstd{(}\hlkwd{rep}\hlstd{(}\hlnum{1}\hlstd{, p))}
  \hlkwa{if}\hlstd{(truncate)\{}
    \hlstd{alpha} \hlkwb{<-} \hlkwd{pnorm}\hlstd{(}\hlkwd{sqrt}\hlstd{(}\hlkwd{log}\hlstd{(n)}\hlopt{/}\hlnum{2}\hlstd{))}
  \hlstd{\}}
  \hlstd{i} \hlkwb{<-} \hlnum{1}
  \hlkwa{while}\hlstd{(i} \hlopt{<} \hlstd{maxiter)\{}
    \hlstd{Finv0} \hlkwb{<-} \hlkwd{approxfun}\hlstd{(}\hlkwd{ecdf}\hlstd{(mm}\hlopt{%*%}\hlstd{beta_eq)(mm}\hlopt{%*%}\hlstd{beta_eq),}
      \hlstd{mm}\hlopt{%*%}\hlstd{beta_eq,} \hlkwc{rule} \hlstd{=} \hlnum{2}\hlstd{)}
    \hlkwa{if}\hlstd{(truncate)\{}
      \hlstd{tFinv} \hlkwb{<-} \hlkwa{function}\hlstd{(}\hlkwc{x}\hlstd{)\{}
        \hlkwd{ifelse}\hlstd{(x} \hlopt{<} \hlnum{1}\hlopt{-}\hlstd{alpha,} \hlkwd{Finv0}\hlstd{(}\hlnum{1}\hlopt{-}\hlstd{alpha),}
        \hlkwd{ifelse}\hlstd{(x} \hlopt{>} \hlstd{alpha,} \hlkwd{Finv0}\hlstd{(alpha),} \hlkwd{Finv0}\hlstd{(x)))}
      \hlstd{\}}
      \hlstd{Finv} \hlkwb{<-} \hlstd{tFinv}
    \hlstd{\}}\hlkwa{else}\hlstd{\{}
      \hlstd{Finv} \hlkwb{<-} \hlstd{Finv0}
    \hlstd{\}}
    \hlstd{mf[,}\hlnum{1}\hlstd{]} \hlkwb{<-} \hlkwd{Finv}\hlstd{(R}\hlopt{/}\hlstd{(n}\hlopt{+}\hlnum{1}\hlstd{))}
    \hlstd{beta_eq} \hlkwb{<-} \hlkwd{Norm}\hlstd{(}\hlkwd{lm}\hlstd{(formula,} \hlkwc{data}\hlstd{=mf, ...)}\hlopt{$}\hlstd{coef)}
    \hlkwa{if}\hlstd{(}\hlkwd{Length}\hlstd{(beta_eq} \hlopt{-} \hlstd{beta_eq0)}\hlopt{<}\hlstd{tol)\{}
      \hlkwa{break}\hlstd{()}
    \hlstd{\}}\hlkwa{else}\hlstd{\{}
      \hlstd{beta_eq0} \hlkwb{<-} \hlstd{beta_eq}
      \hlstd{i} \hlkwb{<-} \hlstd{i}\hlopt{+}\hlnum{1}
    \hlstd{\}}
  \hlstd{\}}
  \hlstd{beta_eq}
\hlstd{\}}

\hlcom{# Spearmax}
\hlstd{spearmax} \hlkwb{<-} \hlkwa{function} \hlstd{(}\hlkwc{formula}\hlstd{,} \hlkwc{data}\hlstd{,} \hlkwc{subset}\hlstd{,} \hlkwc{tol}\hlstd{=}\hlnum{1e-5}\hlstd{,}
  \hlkwc{maxiter}\hlstd{=}\hlnum{100}\hlstd{,} \hlkwc{truncate} \hlstd{=} \hlnum{FALSE}\hlstd{,} \hlkwc{...}\hlstd{)}
\hlstd{\{}
  \hlstd{cl} \hlkwb{<-} \hlkwd{match.call}\hlstd{()}
  \hlstd{mf} \hlkwb{<-} \hlkwd{match.call}\hlstd{(}\hlkwc{expand.dots} \hlstd{=} \hlnum{FALSE}\hlstd{)}
  \hlstd{m} \hlkwb{<-} \hlkwd{match}\hlstd{(}\hlkwd{c}\hlstd{(}\hlstr{"formula"}\hlstd{,} \hlstr{"data"}\hlstd{,} \hlstr{"subset"}\hlstd{,} \hlstr{"weights"}\hlstd{,} \hlstr{"na.action"}\hlstd{,}
      \hlstr{"offset"}\hlstd{),} \hlkwd{names}\hlstd{(mf),} \hlnum{0L}\hlstd{)}
  \hlstd{mf} \hlkwb{<-} \hlstd{mf[}\hlkwd{c}\hlstd{(}\hlnum{1L}\hlstd{, m)]}
  \hlstd{mf}\hlopt{$}\hlstd{drop.unused.levels} \hlkwb{<-} \hlnum{TRUE}
  \hlstd{mf[[}\hlnum{1L}\hlstd{]]} \hlkwb{<-} \hlkwd{quote}\hlstd{(stats}\hlopt{::}\hlstd{model.frame)}
  \hlstd{mf} \hlkwb{<-} \hlkwd{eval}\hlstd{(mf,} \hlkwd{parent.frame}\hlstd{())}
  \hlstd{mm} \hlkwb{<-} \hlkwd{model.matrix}\hlstd{(formula, mf)}
  \hlstd{R} \hlkwb{<-} \hlkwd{rank}\hlstd{(}\hlkwd{model.response}\hlstd{(mf,} \hlstr{"numeric"}\hlstd{),} \hlkwc{na.last}\hlstd{=}\hlstr{'keep'}\hlstd{)}
  \hlstd{p} \hlkwb{<-} \hlkwd{ncol}\hlstd{(mm)}
  \hlstd{ncorfun} \hlkwb{<-} \hlkwa{function}\hlstd{(}\hlkwc{b}\hlstd{)\{}
    \hlopt{-}\hlkwd{cor}\hlstd{(R,} \hlkwd{rank}\hlstd{(mm}\hlopt{%*%}\hlstd{b,} \hlkwc{na.last}\hlstd{=}\hlstr{'keep'}\hlstd{))}
  \hlstd{\}}
  \hlstd{res} \hlkwb{<-} \hlkwd{constrOptim}\hlstd{(}\hlkwd{Norm}\hlstd{(}\hlkwd{rep}\hlstd{(}\hlnum{1}\hlstd{,p)), ncorfun,} \hlkwc{method} \hlstd{=} \hlstr{"Nelder-Mead"}\hlstd{,}
    \hlkwc{ui} \hlstd{=} \hlopt{-}\hlkwd{Norm}\hlstd{(}\hlkwd{rep}\hlstd{(}\hlnum{1}\hlstd{,p)),} \hlkwc{ci} \hlstd{=} \hlopt{-}\hlkwd{Norm}\hlstd{(}\hlkwd{rep}\hlstd{(}\hlnum{1}\hlstd{,p))}\hlopt{*}\hlnum{2}\hlstd{)}
  \hlstd{b} \hlkwb{<-} \hlkwd{Norm}\hlstd{(res}\hlopt{$}\hlstd{par)}
  \hlkwd{names}\hlstd{(b)} \hlkwb{<-} \hlkwd{colnames}\hlstd{(mm)}
  \hlstd{b}
\hlstd{\}}
\end{alltt}
\end{kframe}
\end{knitrout}

\bibliography{rank}

\end{document}